\documentclass[conference,compsoc]{IEEEtran}
\usepackage[utf8]{inputenc} % allow utf-8 input
\usepackage[T1]{fontenc}    % use 8-bit T1 fonts
\usepackage{hyperref}       % hyperlinks
\usepackage{url}            % simple URL typesetting
\usepackage{booktabs}       % professional-quality tables
\usepackage{amsfonts}       % blackboard math symbols
\usepackage{nicefrac}       % compact symbols for 1/2, etc.
\usepackage{microtype}      % microtypography
\usepackage{xcolor}         % colors
\definecolor{mygreen}{RGB}{34,139,34}
\definecolor{myorange}{RGB}{230,120,0}
\usepackage{lipsum}
\usepackage{subcaption}
\usepackage{capt-of}

\usepackage{dsfont}
\usepackage{mleftright}

\usepackage{bm}
\usepackage{amsthm}

\theoremstyle{definition}
\newtheorem{definition}{Definition}[section]

\theoremstyle{plain}

\newtheorem{lemma}{Lemma}[section]
\newtheorem{proposition}{Proposition}[section]

\usepackage{multirow}

\usepackage{amsmath}
\usepackage{pifont}         % fancy checkmark and cross symbols
\usepackage[disable]{todonotes}      % todo notes
\usepackage{soul}
\usepackage{forest}         % visualize case distinction

\usepackage{amssymb}
\DisableLigatures{encoding = T1, family = tt* }

\renewcommand{\vec}[1]{\boldsymbol{\mathbf{#1}}}
\newcommand{\up}[1]{\overline{#1}}
\newcommand{\lo}[1]{\underline{#1}}

\DeclareMathOperator{\net}{\mathcal{N}}  % neural network

\DeclareMathOperator{\gelu}{\mathrm{GELU}}
\DeclareMathOperator{\relu}{\mathrm{ReLU}}
\DeclareMathOperator{\silu}{\mathrm{SiLU}}
\DeclareMathOperator{\elu}{\mathrm{ELU}}
\DeclareMathOperator{\R}{\mathbb{R}}   % real numbers
\usepackage[createShortEnv]{proof-at-the-end}

\usepackage{placeins}
\usepackage{wrapfig}
\usepackage[inline]{enumitem}
\usepackage{algorithm}
\usepackage{algpseudocode}
\usepackage{dsfont}
\usepackage{makecell}

\newtheorem{corollary}{Corollary}[section]



\ifCLASSOPTIONcompsoc
  \usepackage[nocompress]{cite}
\else
  \usepackage{cite}
\fi
\ifCLASSINFOpdf
\else
\fi
\begin{document}

\pagestyle{plain}
%
% paper title
% Titles are generally capitalized except for words such as a, an, and, as,
% at, but, by, for, in, nor, of, on, or, the, to and up, which are usually
% not capitalized unless they are the first or last word of the title.
% Linebreaks \\ can be used within to get better formatting as desired.
% Do not put math or special symbols in the title.
\title{Encrypted Neural Networks without Overflows}

% author names and affiliations
% use a multiple column layout for up to three different
% affiliations
%\author{
%\IEEEauthorblockN{Anonymous Author(s)}
%\IEEEauthorblockA{Paper under double-blind review}
%}
 \author{\IEEEauthorblockN{Philipp Kern}
 \IEEEauthorblockA{Karlsruhe Institute of Technology\\
 Karlsruhe, Germany\\
 Email: philipp.kern@kit.edu}
 \and
 \IEEEauthorblockN{Lorenzo Rovida}
 \IEEEauthorblockA{Polytechnic University of Turin\\
 Turin, Italy\\
 Email: lorenzo.rovida@polito.it}
 \and
 \IEEEauthorblockN{Samuel Teuber}
 \IEEEauthorblockA{Karlsruhe Institute of Technology\\
 Karlsruhe, Germany\\
 Email: teuber@kit.edu}
 \and
 \IEEEauthorblockN{Edoardo Manino}
 \IEEEauthorblockA{The University of Manchester\\
 Manchester, UK\\
 Email: edoardo.manino@manchester.ac.uk}
 \and
 \IEEEauthorblockN{Carsten Sinz}
 \IEEEauthorblockA{Karlsruhe University of Applied Sciences\\
 Karlsruhe, Germany\\
 Email: carsten.sinz@h-ka.de}
 \and
 \IEEEauthorblockN{Alberto Leporati}
 \IEEEauthorblockA{University of Milan-Bicocca\\
 Milan, Italy\\
 Email: alberto.leporati@unimib.it}
 }

% conference papers do not typically use \thanks and this command
% is locked out in conference mode. If really needed, such as for
% the acknowledgment of grants, issue a \IEEEoverridecommandlockouts
% after \documentclass

% for over three affiliations, or if they all won't fit within the width
% of the page (and note that there is less available width in this regard for
% compsoc conferences compared to traditional conferences), use this
% alternative format:
% 
%\author{\IEEEauthorblockN{Michael Shell\IEEEauthorrefmark{1},
%Homer Simpson\IEEEauthorrefmark{2},
%James Kirk\IEEEauthorrefmark{3}, 
%Montgomery Scott\IEEEauthorrefmark{3} and
%Eldon Tyrell\IEEEauthorrefmark{4}}
%\IEEEauthorblockA{\IEEEauthorrefmark{1}School of Electrical and Computer Engineering\\
%Georgia Institute of Technology,
%Atlanta, Georgia 30332--0250\\ Email: see http://www.michaelshell.org/contact.html}
%\IEEEauthorblockA{\IEEEauthorrefmark{2}Twentieth Century Fox, Springfield, USA\\
%Email: homer@thesimpsons.com}
%\IEEEauthorblockA{\IEEEauthorrefmark{3}Starfleet Academy, San Francisco, California 96678-2391\\
%Telephone: (800) 555--1212, Fax: (888) 555--1212}
%\IEEEauthorblockA{\IEEEauthorrefmark{4}Tyrell Inc., 123 Replicant Street, Los Angeles, California 90210--4321}}

% use for special paper notices
%\IEEEspecialpapernotice{(Invited Paper)}

% make the title area
\maketitle

% As a general rule, do not put math, special symbols or citations
% in the abstract
\begin{abstract}
The popular Cheon-Kim-Kim-Song (CKKS) scheme enables efficient private inference in neural networks by evaluating them on encrypted data. Since CKKS only supports addition, multiplication, and array rotation operations, turning neural networks into CKKS circuits requires approximating all activation functions (e.g. ReLU) with polynomials over fixed input ranges. Traditionally, these ranges are estimated via heuristic sampling techniques.
%
%In this paper, we demonstrate for the first time that CKKS circuits are vulnerable to overflow events, where legitimate inputs exceed the tolerances of the circuit and cause corrupted outputs. Furthermore, we propose an efficient algorithm to find such inputs in the (visual) vicinity of benign inputs and demonstrate that it can, e.g., lead to corrupted outputs for 47\% of CIFAR10 images processed by a CKKS circuit. 
%%%%%%%%%%%%%%%%%%%%%%%%%%%%%%%%%%%%%%%%%%%%%%%%%%%%%%%%%%%%%%%%%%%
In this paper, we empirically demonstrate that sampled ranges leave the CKKS network exposed to overflow events, i.e. legitimate inputs that exceed the ranges of the circuit and cause corrupted outputs, and propose an improved design technique that completely avoids them.
%
% It is known that sampled ranges leave the CKKS network exposed to overflow exposed to overflow events, i.e. legitimate inputs that exceed the ranges of the circuit and cause corrupted outputs, in this paper we propose an improved design technique that completely avoids them.
%
First, we propose a practical algorithm to induce overflow events in a CKKS circuit, i.e. we automatically craft legitimate inputs that exceed the design ranges within the circuit and cause corrupted outputs. As an example, our algorithm can quickly corrupt around 47\% of benign inputs in a CIFAR10 neural network.
Second, we propose an end-to-end formal verification approach that constructs executable CKKS circuits. These circuits are guaranteed to be overflow-free and come with certified error bounds on the output difference w.r.t. the original neural network.
Our experiments demonstrate that our approach supports a larger set of neural network architectures than prior certified techniques, while also delivering executable CKKS artifacts.
%%%%%%%%%%%%%%%%%%%%%%%%%%%%%%%%%%%%%%%%%%%%%%%%%%%%%%%%%%%%%%%%%%%%%
% Our approach extends beyond simple ReLU-based fully connected networks supported by prior certified techniques  and handles convolutions and GELU activations, where polynomial approximation is more effective.
% Our experiments demonstrate that this enables us to compute significantly tighter certified error bounds on GELU neural networks than for previously supported ReLU networks of similar size.

%
% Second, we propose...
%
%To completely eliminate overflows, we propose a formal verification technique that computes certified bounds on the ranges of all neurons in the network. 
%Our overflow-free solution is compatible with most CKKS-based frameworks,
%PROPOSAL: Our overflow-free solution is compatible with fully connected and convolution CKKS-based networks,
%as it allows to simply substitute heuristically-constructed polynomials with rigorously designed ones.
\end{abstract}
%\todo[inline]{Given Theorem 2 from \cite{10.1007/978-3-031-15802-5_20}, we can write down the theorem for overflow free ckks networks. (in abstract terms using heuristics, but still nice to present)}

% no keywords

% For peer review papers, you can put extra information on the cover
% page as needed:
% \ifCLASSOPTIONpeerreview
% \begin{center} \bfseries EDICS Category: 3-BBND \end{center}
% \fi
%
% For peerreview papers, this IEEEtran command inserts a page break and
% creates the second title. It will be ignored for other modes.
\IEEEpeerreviewmaketitle

\section{Introduction}
\label{sec:introduction}
Private inference schemes based on fully-homomorphic encryption (FHE) are designed to delegate the bulk of the computation to a third-party, without revealing the data itself~\cite{pulido2021privacy}. This is especially relevant for the so-called machine learning as a service (MLaaS) business model, where clients buy access to a machine learning model and run it on third-party servers~\cite{ribeiro2015mlaas}. An ideal FHE private inference scheme would encrypt the client data before transmission to the server, run the full inference over encrypted data, and only decrypt the output once it is transmitted back to the client. % (see Figure \ref{fig:oflow_atk}\todo{LR: The figure does not really relate with this sentence}).
That way, the third-party server does never access the plaintext data, thus guaranteeing its confidentiality.%which might contain personal %or sensitive
%information.\footnote{The honest-but-curious threat model for the third-party server is a staple of the private inference literature.}

\begin{figure}[h]
  \begin{center}
    \includegraphics[width=\linewidth]{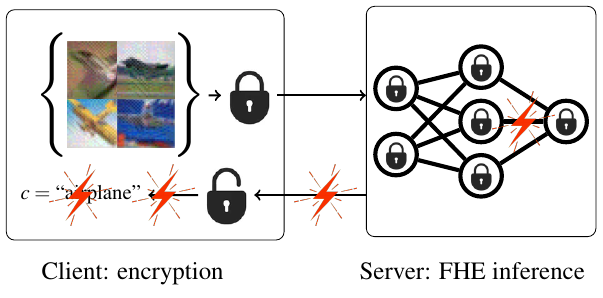}
  \end{center}
\caption{%
FHE allows neural network inference on encrypted data.
%Benign inputs (e.g.\ perturbed CIFAR images) can exceed expected ranges of the FHE circuit, corrupting the output.}
However, some legitimate inputs can exceed the expected ranges of the FHE circuit, thus corrupting the output.}
\label{fig:oflow_atk}
\end{figure}
Several FHE schemes have been designed to realize this vision, including BGV~\cite{BGV12}, BFV~\cite{Bra12,FV12} and TFHE~\cite{CGGI16}.
Among them, the CKKS scheme~\cite{CKKS17} has emerged as the de facto standard for efficiently implementing private neural inference, as it allows to natively encode vectors of complex numbers and perform massive arithmetic operations in parallel. Recent work demonstrates that executing deep convolutional networks~\cite{Lee2023convolutional} and large language models~\cite{10.1145/3643651.3659893,LLAMACKKS} with CKKS is possible.
However,
%At the same time,
the CKKS scheme imposes severe restrictions on the arithmetic operations available over the ciphertext: it only supports addition, multiplication and array rotation.
Thus,
%For this reason,
existing CKKS inference frameworks must replace common %all
non-linear activation functions like the ReLU with (%potentially 
high degree) polynomials within a certain interval range~\cite{Lee2022minimax,Ebel2025orion}.

\begin{figure*}[t]
 \centering
 \includegraphics[width=0.94\textwidth]{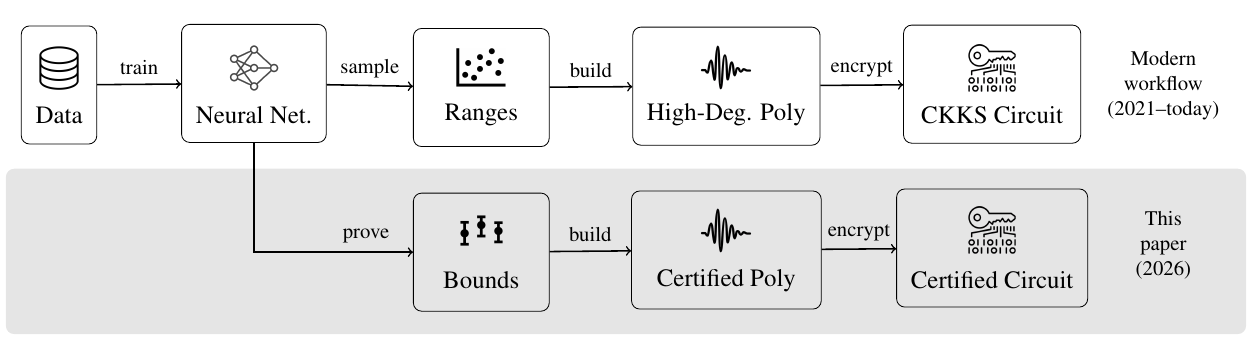}
 \caption{Designing CKKS neural networks requires approximating activations within their expected ranges. This paper avoids overflow events by computing provable bounds and certified polynomials.}
\label{fig:lit_map}
\end{figure*}

%\looseness=-1
%In this paper, we demonstrate a major weakness in modern CKKS designs for neural networks (NNs):\todo{Fix this sentence (the weakness is known)} 
However, there is a major weakness in modern CKKS designs for neural networks (NNs):
there exists a sizeable set of legitimate inputs that make the CKKS network arbitrarily diverge from its reference plaintext implementation. This is because the polynomial approximations of non-linear activation functions are only tight for a limited range of values~\cite{Manino2023fhe,Kern2025fhe}.
A specific input needs only to exceed the range of a \emph{single} activation to cause a cascade of overflows through the CKKS circuit and yield unrecoverable output values.
%need to only exceed the expected range of a single activation to cause a cascade of overflows through the CKKS circuit, which eventually leads to unrecoverable output values.
We call this phenomenon an \textit{overflow event} and note that it breaks the property of availability for the CKKS network: some inputs, no matter how reasonable-looking, will not produce a usable inference result %\footnote{We understand overflow attacks similar to adversarial attacks, not in the sense of cryptographic attacks.}
(see Figure \ref{fig:oflow_atk}).
%
%\alberto{Related with the previous observation: If I ignore the reference to FHE networks instead of circuits, up to now it seems that we present a result concerning FHE circuits. So why are we submitting it to a conference which is about neural networks and the like?}

To address this challenge, we introduce a framework for the \emph{certified} design of FHE NNs.
%\todo{Give it a catchy name? E.g. CertiFHENN. Or more simply FHENN, or HomoNet :-)} 
Our framework guarantees the absence of overflow events, without any noticeable increase in the size of the FHE circuit. %\todo{@Lorenzo: check this}
%while keeping the size of the FHE circuit as small as possible.
In detail, we propose the following contributions:
% More in detail,
%We propose the following contributions to the state of the art:
\begin{itemize}[leftmargin=2em]
    \item We present an efficient algorithm to find inputs that exceed the expected ranges of NNs with polynomial activations. With it, we demonstrate that the modern design of CKKS NNs is vulnerable to overflows. Our empirical result shows that even small perturbations of benign inputs can cause up to $47\%$ of unrecoverable outputs for a CIFAR10 network.
    %\item We show that the modern approach of approximating activation functions with polynomials has two issues. First, it makes the network vulnerable to adversarial perturbations, as they can make the internal neurons overflow their expected ranges. Second, it requires exceptionally high degrees ($d>100$) to compensate for the approximation error on non-analytic/non-smooth activation functions such as ReLU.
%
    \item To mitigate this, we introduce a differential verification framework that can programmatically compute certified bounds on the value range of all neurons in the NN and on the approximation error of the output. The framework can handle both analytic (e.g. GELU) and non-analytic activation functions (e.g. ReLU), fully-connected and convolutional layers.
    \item We give theoretical proofs that our bounds are sound for any possible input, i.e. both benign and adversarial, and thus can be used at design time to eliminate the possibility of overflows. We also show that NNs with analytical activation functions (e.g. GELUs) yield an asymptotically smaller approximation error than NNs with non-analytical ones (e.g. ReLUs). This gives theoretical grounding to a recent empirical observation reported by Ebel~\emph{et al.}~\cite{Ebel2025orion}. %\todo{Do we want to claim retraining? Need to show that we are better than standard ReLU networks.}.
    \item We further improve the design of CKKS NNs by allowing for the polynomial activations of different neurons in the same layer to be different, as opposed to uniform approximations in existing work.
    We thus generalize the Paterson-Stockmeyer algorithm~\cite{PS73,CCS19} from evaluating (ciphertext $\times$ constant) to (ciphertext $\times$ plaintext). Our experiments show that a heterogeneous design can achieve at least a $4.7\times$ reduction in error compared to the uniform one (as shown in Figure \ref{fig:ablation}).%\todo{LR: works for me, double check :) We don't show it via OpenFHE but just on $f_{\pi}$ Right!}
    %\todo{LR:Remove this $\%$? We'd need to perform comparisons with other works to justify the claim. We might say that to the best of our knowledge this has never been done in literature?}

    %\item \textcolor{orange}{We improve the IND-CPA$^\text{D}$ security \cite{LM21} of CKKS-based neural networks by allowing a simpler study of noise growth during circuit evaluation. The IND-CPA$^\text{D}$ security of a CKKS-based circuit can be improved by having strong estimates on the error contained in the final result \cite{JCMNT25}; our work helps towards that direction as it allows us to construct a CKKS circuit whose output error is not increased by polynomial approximations of activation functions, unlike what typically occurs in standard CKKS-based machine learning.}\todo{Discuss and polish}

    %LR: put this below (after experiments)
    
\end{itemize}

%\todo{Supplemental material or anonymous repo?}
We provide an open-source implementation\footnote{Refer to \url{https://github.com/lorenzorovida/encrypted-neural-networks-without-overflows}}
%\alberto{Modify the sentence to include the anonymous repo?} 
of our framework, based on OpenFHE~\cite{OpenFHE}, which includes construction of certified polynomials and conversion to CKKS NNs.

\subsubsection*{Overview}
Section \ref{sec:background} provides the necessary background on FHE (in particular CKKS) as well as NNs, their verification and their polynomial approximation.
Section \ref{sec:overflow} introduces the notion of overflow events and presents an efficient algorithm for the discovery of overflowing inputs.
Section \ref{sec:cert_design} describes our methodology for the certified design of overflow-free CKKS circuits.
Section \ref{sec:poly-approx-smooth} provides a theoretical argument for the usage of $\gelu{}$ activation functions  and Section \ref{sec:concrete-ckks} describes how our approximations can be efficiently implemented in CKKS circuits.
We evaluate our approach in Section \ref{sec:experiments}.

\section{Background}
%\todo[inline]{LR: todo :)}
State-of-the-art FHE NNs, especially those based on the CKKS scheme, are typically constructed in a multi-stage process, as shown in the upper part of Figure~\ref{fig:lit_map}: a regular neural network $f(\mathbf{x})$ is first approximated as a high-degree polynomial $f_\pi(\mathbf{x})$, after which the corresponding CKKS circuit $f_\text{CKKS}(\mathbf{x})$ can be built. Throughout the process, several sources of error are introduced:
\begin{equation}
\begin{split}
\label{eq:error_sources}
    f_\text{CKKS}(\mathbf{x})&=f_\pi(\mathbf{x}) + {\bm e}_\text{CKKS}(\mathbf{x})\\&=f(\mathbf{x})+{\bm e}_{\pi}(\mathbf{x})+{\bm e}_\text{CKKS}(\mathbf{x}) \enspace.
    \end{split}
\end{equation}
where ${\bm e}_{\pi}(\mathbf{x})=f_{\pi}(\mathbf{x})-f(\mathbf{x})$ is the \emph{deterministic} polynomial approximation error stemming from the approximate representation of activation functions in $f_\pi$, while ${\bm e}_{\text{CKKS}}(\mathbf{x})$ is the stochastic error stemming from the CKKS encryption.

In the remainder of the section, we will characterise these sources of errors, from the encryption itself to polynomial approximation and range estimation. %by going backwards through Figure \ref{fig:lit_map}: Starting from the encryption and following up with polynomial approximation and range estimation.\todo{Paragraph  not very intuitive}

%\todo[inline]{EM: Move Eq. 1 here. Say explicitly we are going through Figure 2 backwards.\\
%LR: Done, what you think? Perhaps we can stress the circuit->high deg poly->ranges flux further (moved to section 2.3)}

\label{sec:background}

\subsection{Fully Homomorphic Encryption}
\label{sec:fhe-intro}

We begin by introducing some abstract definitions of Fully Homomorphic Encryption (FHE) schemes reported in the literature~\cite{10.1007/978-3-031-15802-5_20,10.62056/ayl83z10k} and extended to our needs.
\begin{definition}[FHE Scheme]
A (public-key) homomorphic encryption scheme with plaintext space
$\mathcal{M}$, ciphertext space $\mathcal{C}$, public key space $\mathcal{PK}$,
secret-key space $\mathcal{SK}$, and space of evaluatable circuits
$\mathcal{L}$ is a tuple of four probabilistic polynomial-time algorithms:
\begin{align*}
\mathsf{KeyGen} &: 1^{\mathbb{N}} \rightarrow \mathcal{PK} \times \mathcal{SK}&
\mathsf{Enc} &: \mathcal{PK} \times \mathcal{M} \rightarrow \mathcal{C}\\
\mathsf{Eval} &: \mathcal{PK} \times \mathcal{L} \times \mathcal{C}&
\mathsf{Dec} &: \mathcal{SK} \times \mathcal{C} \rightarrow \mathcal{M}
\rightarrow \mathcal{C}
\end{align*}
\end{definition}
In particular, in this work we will be dealing with \emph{approximate-FHE}, i.e., such that $\mathsf{Dec}(\mathsf{Enc}(m)) \approx m$ for any $m \in \mathcal{M}$. Since the result of a decryption is approximate, standard correctness notions cannot be applied here. We thus first introduce a definition related to the impact of such approximation.
\begin{definition}[Ciphertext Error \cite{10.1007/978-3-031-15802-5_20}]
    Let $\prod = (\mathsf{KeyGen}, \mathsf{Enc}, \mathsf{Dec}, \mathsf{Eval})$ be an
FHE scheme with message space $\mathcal{M} \subseteq \tilde{\mathcal{M}}$, which is a normed space with norm
$\|\cdot\|: \tilde{\mathcal{M}} \rightarrow \mathbb{R}_{\geq 0}$. For any ciphertext $\mathsf{ct}$, secret key $\mathsf{sk}$, and message $m$, the
\emph{ciphertext error} of $(\mathsf{ct}, m, \mathsf{sk})$ is defined to be
\[
\mathsf{Error}(\mathsf{ct}, m, \mathsf{sk}) = \|\mathsf{Dec}_{\mathsf{sk}}(\mathsf{ct}) - m\|.
\]
\end{definition}
Note that we will conveniently require for $m_1,\dots, m_k$ to be bounded, and this can be captured in the above definition by choosing $\mathcal{M}$ to be a set of bounded norm. It turns out that this definition is compatible with neural network verification (see Section \ref{sec:background_nnv}) which also requires inputs  $m \in \mathcal{M}$ with bounded norm.

%It is sufficient for the sake of our work to assume working with a definition of correctness that is static, i.e., it only depends on a class of circuits $\mathcal{L}_k$ and it can be (publicly) precomputed offline, given a circuit and a set of bounded inputs. 
For this work, it suffices to assume static correctness, i.e., we assume a statically computable function $\mathsf{Estimate}$ that bounds the ciphertext error after circuit execution:

%\todo[inline]{EM: the purpose of Def. \ref{def:correctness} remains  obscure to me. I added a sentence ad the end of Sec. 3 that tries to close the loop, let me know!} % thanks, now at least we use the definition somewhere!

\begin{definition}[Static Approximate Correctness \cite{10.1007/978-3-031-15802-5_20}]
Let $\prod = (\mathsf{KeyGen}, \mathsf{Enc}, \mathsf{Dec}, \mathsf{Eval})$ be an
FHE scheme with message space $\mathcal{M} \subseteq \tilde{\mathcal{M}}$, which is a normed space with norm
$\|\cdot\|: \tilde{\mathcal{M}} \rightarrow \mathbb{R}_{\geq 0}$.
Let $\mathcal{L}$ be a space of circuits, $\mathcal{L}_k \subseteq \mathcal{L}$ the subset of parity $k$ circuits, and let
$\mathsf{Estimate}: \sqcup_{k \in \mathbb{N}} \mathcal{L}_k \times \mathbb{R}_{\ge 0}^k \rightarrow \mathbb{R}_{\ge 0}$ be an efficiently computable function. We
call the tuple $\tilde{\prod} = (\prod, \mathsf{Estimate})$ a \emph{statically approximate FHE scheme} if for
all $k \in \mathbb{N}$, for all $\mathcal{C} \in \mathcal{L}_k$, for all $(\mathsf{pk}, \mathsf{sk}) \gets \mathsf{KeyGen}(1^\kappa)$, if $\mathsf{ct}_1, \dots, \mathsf{ct}_k$ are the encryptions of
$m_1, \dots, m_k$ with $\mathrm{Error}(\mathsf{ct}_i, m_i, \mathsf{sk}) \leq t_i$, then
\[
\begin{split}
\mathsf{Error}(\mathsf{Eval}_{\mathsf{pk}}(C, \mathsf{ct}_1, \dots, \mathsf{ct}_k), C(m_1, \dots, m_k), \mathsf{sk}) \leq\\ \mathsf{Estimate}(C, t_1, \dots, t_k).
\end{split}
\]
\label{def:correctness}
\end{definition}
%\todo[inline]{LR: How does it deal with plaintext overflow though? I asked the authors}
At the core of this definition is the $\mathsf{Estimate}$ function, which can be computed using different approaches \cite{10.1007/978-3-031-15802-5_20,10.62056/ayl83z10k}.
In practice, it acts as a tool to measure how far the CKKS circuit diverges from the plaintext one.
%Note that, at the core of such definition lie the $\mathsf{Estimate}$ function, which can be computed using different approaches as discussed in \cite{10.1007/978-3-031-15802-5_20,10.62056/ayl83z10k}. In practice, it acts as a tool to measure how far the CKKS circuit is behaving with respect to the plain one.
%We note that different instantiations/encodings of CKKS can make the estimation easier, e.g., bivariate-CKKS \cite{10.1007/978-981-96-0875-1_6} seems like a very interesting direction in this sense.
%\todo{do we keep this in the appendix?}% added by ST
%\todo[inline]{LR: todo}
%We refer the reader to Appendix \ref{sec:generalfhe} for a more general introduction to the available FHE schemes

\subsection{The CKKS Scheme}
The Cheon--Kim--Kim--Song (CKKS) scheme~\cite{CKKS17} currently represents the standard framework to efficiently handle complex (and real) arithmetic in FHE. 
Similar to other FHE schemes, it is lattice-based and has a security reduction to the hardness of the Ring Learning with Errors (RLWE) problem \cite{LPR10}. 
This is a structured version of the original Learning with Errors (LWE) \cite{Reg09} introduced for efficiency reasons. 

%In this work, we employ an optimized version of the CKKS scheme which uses the double Chinese Remainder Theorem (DCRT) representation of polynomials \cite{KPP22}, implemented in the OpenFHE library \cite{OpenFHE}. 

\subsubsection{Basics of the Scheme}
Let $N = 2^k$ for some integer $k > 0$ and let $Q > 0$ be some large modulus.
Let $\mathsf{DFT}:\mathbb{R}[X]/(X^N+1) \rightarrow \mathbb{C}^{N/2}$ be a discrete Fourier transform (DFT) defined as ${\bm m}(X) \mapsto ({\bm m}(\zeta^{5^i}))_{0 \leq i < N/2}$,\footnote{$(1, 5, \dots, 2^{N/2} - 3)$ is a cyclic subgroup of the multiplicative group $\mathbb{Z}_{N}^\times$ generated by 5. Rotations are thus enabled via $X \mapsto X^5$ \cite{10.1007/978-3-319-78381-9_14}.} and let $\mathsf{iDFT}: \mathbb{C}^{N/2} \rightarrow \mathbb{R}[X]/(X^N+1)$ be its inverse, with $\zeta$ being a $2N$-th primitive root of unity.
The cleartext space of CKKS is $\mathbb{C}^{N/2}$, the plaintext space (i.e., after the encoding) is $\mathcal{R} := \mathbb{Z}[X]/(X^N+1)$, and the ciphertext space is $\mathcal{R}_{Q}^2$, where $\mathcal{R}_Q$ is the quotient ring  $\mathcal{R}_Q = \mathcal{R}/Q\mathcal{R}$.

%\todo[inline]{CS: What about moving Def. 2.4 here?
%The sentence before that def. about LWE could then perhaps be dropped or put at the end of this new paragraph.
%}
\begin{definition}[Ring Learning with Errors (RLWE) \cite{LPR10}] 
Let $\chi_s, \chi_e$ be distributions over $\mathcal{R}_Q$.% \todo{is the modulus missing here? Uhmm we may add it, typically those distributions are very small (e.g. from [-1, 0, 1])}. 
The goal of the RLWE problem is to distinguish $({\bm a}, {\bm x})$ from $({\bm a}, {\bm a}{\bm s} + {\bm e})$ for uniformly random ${\bm a}, {\bm x} \in \mathcal{R}_Q$, ${\bm e} \leftarrow \chi_{\bm e}$, and ${\bm s} \leftarrow \chi_{\bm s}$. The RLWE$_{N, Q, \chi_e, \chi_s}$ assumption is that solving the RLWE$_{N, Q, \chi_e, \chi_s}$ problem is computationally unfeasible.
%\todo{CS: Is it really $\chi_e$ twice and no $\chi_s$ in RLWE?}
\end{definition}

In order to make the cleartext space compatible with RLWE encryption (which takes a message in $\mathcal{R}$ and returns a ciphertext in $\mathcal{R}_{Q}^2$%\todo{why is there an $\ell$ index for the $Q$?})
, the CKKS scheme uses (the inverse of) a discrete Fourier transform encoding that moves a plaintext from $\mathbb{C}^{N/2}$ to $\mathcal{R}$ as follows:
\[
\begin{split}
\underbrace{\mathbb{C}^{N/2}}_{\text{Cleartext}} \xrightarrow[\text{Encoding}]{\;\text{Rounded iDFT}\;}  \underbrace{\mathcal{R}}_{\text{Plaintext}}
\xrightarrow[\text{Encryption}]{\,\;\;\;\text{RLWE}\;\;\;\,} \underbrace{\mathcal{R}_Q^2}_{\text{Ciphertext}},
\end{split}
\]
where the rounding follows ideas from Lyubashevsky, Peikert and Regev~\cite[Section 2]{LPR10}. As it will be shown later, this rounding operation implies that we are required to scale the message by some large factor $\Delta > 0$ to preserve the most significant digits.
In practice, numbers behave similarly to a fixed-point system. The scaling factor $\Delta$ is the key parameter governing computational precision: it yields greater accuracy but comes at the cost of increasing the ciphertext scale from $\Delta$ to $\Delta^2$ after each multiplication.

To tackle this, CKKS introduces a moduli chain $Q_\ell = \prod_{i=0}^\ell q_i$, with $q_i \approx \Delta$, which enables a technique called \emph{rescaling}: after a multiplication (whose result has scale $\Delta^2$), the scale is reduced back to $\Delta$ by dividing by $q_i$, where $i$ is the current multiplicative level. However, each rescaling operation consumes one modulus from the chain, so the depth of circuits that can be evaluated is bounded by the number of available $q_i$ (number of levels $\ell$) --- unless additional techniques such as bootstrapping \cite{10.1007/978-3-319-78381-9_14} are employed.
The homomorphisms induced by ($\mathsf{i}$)$\mathsf{DFT}$
\[
\begin{split}
a + b &= \mathsf{DFT}(\mathsf{iDFT}(a) + \mathrm{iDFT}(b)) \quad \text{and}\\ \quad a \odot b &= \mathsf{DFT}(\mathsf{iDFT}(a) \cdot \mathsf{iDFT}(b)),
\end{split}
\]
where $\odot$ is the slot-wise multiplication, naturally induce a Single-instruction multiple data (SIMD) computational paradigm on the plaintext space by simply adding/multiplying ciphertexts in the iDFT domain.
We can summarize the main operations of the scheme in what follows.
\begin{itemize}
	\item $\mathsf{CKKS.Setup}$$(1^\lambda, N, \ell)$. Given a security parameter
	$\lambda$, a ring dimension $N$, and a multiplicative circuit depth $\ell$,
	choose a modulus chain $Q_\ell = \prod_{i=0}^{\ell} q_i$, a secret key distribution
	$\chi_s$, an error distribution $\chi_e$, and a scaling factor
	$\Delta \in \mathbb{Z}_{>0}$. Output the public parameter tuple
	\[
	{\bm t} := \bigl(N,\; \Delta,\; Q_\ell,\; \chi_s,\; \chi_e\bigr).
	\]
	
	\item $\mathsf{CKKS.KeyGen}$$({\bm t})$. Given the public parameter tuple
	$\mathbf{t}$, sample ${\bm s} \sim \chi_s$ and $\mathbf{a} \leftarrow \mathcal{R}_{Q_\ell}$ uniformly, and
	$\mathbf{e} \sim \chi_e$. Output the secret and public keys
	\[
    \begin{split}
	\mathsf{sk} := {\bm s} \in \mathcal{R}, \quad
	\mathsf{pk} := \bigl(-{\bm a}\cdot {\bm s} + {\bm e},\; {\bm a}\bigr) \in \mathcal{R}_{Q_\ell}^2.
    \end{split}
	\]
	Note that $\mathsf{pk}_0 + \mathsf{pk}_1 \cdot {\bm s} = {\bm e} \approx 0$, so the
	public key is a noisy encryption of zero under ${\bm s}$.
	
	\item $\mathsf{CKKS.Encode}$$(\mathbf{m})$. Return:
	$
	\bigl\lfloor \Delta \cdot \mathsf{iDFT}({\bm m}) \bigr\rceil
	\;\in\; \mathcal{R}.
	$
	\item $\mathsf{CKKS.Decode}$$({\bm p})$. Return:
	$
	1/\Delta\cdot\mathsf{DFT}({\bm p})
	\;\in\; \mathbb{C}^{N/2}.
	$
	
	\item $\mathsf{CKKS.Encrypt}$$({\bm p})$. Given a plaintext
	${\bm p} \in \mathcal{R}_{Q_\ell}$ and the public key $pk$, sample a
	random mask ${\bm u} \sim \chi_s$ and noise terms
	${\bm e}_0, {\bm e}_1 \sim \chi_e$. Output the ciphertext
	\[
	\bigl(
	[pk_0\,{\bm u} + {\bm e}_0 + \mathbf{p}]_{Q_\ell},\;
	[pk_1\,{\bm u} + {\bm e}_1]_{Q_\ell}
	\bigr)
	\;\in\; \mathcal{R}_{Q_\ell}^2.
	\]
	
	\item $\mathsf{CKKS.Decrypt}$$({\bm c})$. Given a ciphertext
	${\bm c} = (c_0, c_1) \in \mathcal{R}_{Q_\ell}^2$ and the secret key
	${\bm s}$, compute the inner product $c_0 + c_1\,\mathbf{s}$. Output
	the approximate plaintext
	\[
	\begin{split}
		c_0 + c_1\,\mathbf{s}
		&= \bigl(pk_0\,{\bm u} + {\bm e}_0 + {\bm p}\bigr)
		+ \bigl(pk_1\,{\bm u} + {\bm e}_1\bigr){\bm s} \\
		%&= \bigl(-{\bm a}{\bm s} + {\bm e}\bigr){\bm u}
		%+ {\bm e}_0 + {\bm p}
		%+ {\bm a}\,{\bm u}\,{\bm s}
		%+ {\bm e}_1\,{\bm s} \\
		&= \underbrace{{\bm e}\cdot {\bm u}
			+ {\bm e}_0
			+ {\bm e}_1\,{\bm s}}_{\text{small noise}}
		+ {\bm p}
		\;\approx\; {\bm p}
		\;\in\; \mathcal{R}_{Q_\ell}.
	\end{split}
	\]
Notice that the term indicated above with \emph{small noise} is precisely the ${\bm e}_\text{CKKS}({\bm p})$ term abstractly introduced in Eq.\eqref{eq:error_sources}.
\end{itemize}
Encryption and decryption follow from standard RLWE \cite{LPR10}. Notice that the CKKS noise is mainly given by two factors: the rounding performed in the encoding and the error injected during encryption.

%\todo[inline]{EM: call back to the ${\bm e}_\text{CKKS}(\mathbf{x})$ term in Equation \ref{eq:error_sources}? Good idea!}

\subsubsection{Homomorphic Operations}
\label{subsubsec:background:ckks:ops}
%The CKKS scheme enables a limited number of homomorphic operations \cite{CKKS17}, namely:
The CKKS scheme only supports a limited number of homomorphic operations~\cite{CKKS17}, namely:
%\todo{ST: Before we strongly emphasize that CKKS only supports limited operations... Can we exhaustively list here which other operations it supports besides the ones mentioned here or are there still too many?}
\begin{itemize}
	\item $\mathsf{CKKS.Add}$$({\bm c}, {\bm c}')$. Given two ciphertexts
	${\bm c}, {\bm c}' \in \mathcal{R}_{Q_\ell}^2$, output the
		component-wise sum
		\[
		{\bm c}_{\text{add}}
		\;:=\;
		{\bm c} + {\bm c}'
		\;=\;
		(c_0 + c_0',\; c_1 + c_1')
		\;\in\; \mathcal{R}_{Q_\ell}^2,
		\]
		satisfying
		$\text{Decrypt}({\bm c}_{\text{add}})
		\approx \text{Decrypt}({\bm c}) + \text{Decrypt}({\bm c}')$.
		
		\item $\mathsf{CKKS.Mult}$$({\bm c}, {\bm c}')$. Given two ciphertexts
		${\bm c}, {\bm c}' \in \mathcal{R}_{Q_\ell}^2$, compute the
		degree-2 tensor product
		\[
		(c_0 c_0',\; c_0 c_1' + c_1 c_0',\; c_1 c_1')
		\;\in\; \mathcal{R}_{Q_\ell}^3,
		\]
		and apply $\text{KeySwitch}(\,\cdot\,, \mathsf{sk}^2 \to \mathsf{sk})$ to reduce back
		to two components. Output
		\[
		{\bm c}_{\text{mul}} \;\in\; \mathcal{R}_{Q_\ell}^2,
		\]
		satisfying
		$\mathsf{Decrypt}({\bm c}_{\text{mul}}) \approx
		\mathsf{Decrypt}({\bm c}) \cdot \mathsf{Decrypt}({\bm c}')$.
		Note that the resulting scaling factor is $\Delta^2$; a subsequent call
		to $\mathsf{Rescale}$ is typically needed to restore it to $\Delta$.
		\item $\mathsf{CKKS.KeySwitch}$$({\bm c}, \mathsf{evk})$. Given a ciphertext
		${\bm c} = (c_0, c_1, c_2) \in \mathcal{R}_{Q_\ell}^3$ encrypted
		under $\mathsf{sk}^2$ and an evaluation key $\mathsf{evk}$ --- a public encryption of
		$\mathsf{sk}^2$ under $\mathsf{sk}$ --- decompose $c_2$ into a base-$P$ representation (as in \cite{GHS12})
		$c_2 = \sum_j c_2^{(j)} P^j$ and absorb it into the first two
		components using $\mathsf{evk}$. Output the key-switched ciphertext
		\[
		\bigl(c_0 + \textstyle\sum_j c_2^{(j)}\, \mathsf{evk}_0^{(j)},\;
		c_1 + \textstyle\sum_j c_2^{(j)}\, \mathsf{evk}_1^{(j)}\bigr)
		\;\in\; \mathcal{R}_{Q_\ell}^2,
		\]
		which encrypts the same plaintext under $\mathsf{sk}$ with a controlled increase
		in noise.
        %\todo{Is it clear what $P$ of base $P$ is here? Is it clear what an evaluation key is?} % Added by ST
        %I added a reference where the base P representation is justified
		
		\item $\mathsf{CKKS.Rescale}$$({\bm c})$. Given a ciphertext
		${\bm c} \in \mathcal{R}_{Q_\ell}^2$ with scaling factor $\Delta^2$,
		drop the last modulus $q_\ell$ from the chain and divide. Output
		\[
		\bigl\lfloor \Delta^{-1} \cdot {\bm c} \bigr\rceil
		\;\in\; \mathcal{R}_{Q_{\ell-1}}^2,
		\]
		which restores the scaling factor to $\Delta$ at the cost of consuming
		one level of the modulus chain.
		
		\item $\mathsf{CKKS.Rotate}$$({\bm c}, i)$. Given a ciphertext
		${\bm c} \in \mathcal{R}_{Q_\ell}^2$ and a rotation index
		$i \in \mathbb{Z}$, apply the Galois automorphism
		$\sigma_i \colon X \mapsto X^{5^i \bmod 2N}$ component-wise to obtain
		$\sigma_i({\bm c}) \in \mathcal{R}_{Q_\ell}^2$, and apply
		$\mathsf{KeySwitch}(\,\cdot\,, \sigma_i(\mathsf{sk}) \to \mathsf{sk})$ using the
		corresponding rotation key. Output
		\[
		{\bm c}_{\text{rot}} \;\in\; \mathcal{R}_{Q_\ell}^2,
		\]
		which satisfies
		$\mathsf{Decrypt}({\bm c}_{\text{rot}}) \approx
		\rho^i\!\left(\mathsf{Decrypt}({\bm c})\right)$,
		where $\rho^i$ denotes a cyclic shift of the plaintext slot vector by
		$i$ positions.
\end{itemize}
%\todo[inline]{Add a few words about operations runtime and bootstrapping}
We summarize the noise growth and the runtime of the available operations in Table \ref{tab:ckksoperations}.
\begin{table}[t]
\centering
\normalsize
\setlength{\tabcolsep}{4pt} % Default value: 6pt
\newcommand{\lowSym}{$\bullet$}
\newcommand{\mediumSym}{$\bullet\bullet$}
\newcommand{\highSym}{$\bullet\bullet\bullet$}
\begin{tabular}{llll}
\toprule
Operation & Noise Growth & Latency & Levels Consumed \\
\midrule
$\mathsf{Add}({\bm c}_1, {\bm c}_2)$    & \lowSym  & \lowSym  & unchanged \\
$\mathsf{Sub}({\bm c}_1, {\bm c}_2)$    & \lowSym  & \lowSym  & unchanged \\
$\mathsf{Mul}({\bm c}_1, {\bm c}_2)$    & \highSym      & \highSym & $1$      \\
$\mathsf{Mul}({\bm c}, {\bm p})$        & \mediumSym & \lowSym  & $1$      \\
$\mathsf{Rotate}({\bm c}, i)$     & \lowSym & \highSym      & unchanged \\
\bottomrule
\end{tabular}
\caption{Overview on arithmetic and rotation operations in CKKS.
\lowSym~low,
\mediumSym~medium,
\highSym~high.}
\label{tab:ckksoperations}
\end{table}
Nevertheless, we refer the interested reader to the RNS-CKKS implementation by Kim, Papadimitriou and Polyakov \cite{KPP22} used in this work.

\subsection{Polynomial Networks}
\label{sec:poly_net}

\newcommand{\inputSpace}{\ensuremath{\mathcal{I}}}
\newcommand{\errorVal}{\ensuremath{\bot}}
\newcommand{\inputVal}{\ensuremath{\mathbf{x}}}
\newcommand{\errorBound}{\ensuremath{\mathcal{E}}}
\newcommand{\dataset}{\ensuremath{\mathcal{D}}}
\newcommand{\noiseVector}{\ensuremath{\boldsymbol{\delta}}}
\newcommand{\perturbedVersion}{\ensuremath{\gamma}}

%\label{def:arith_circuit}
%\label{def:ckks_circuit}

Due to the restricted set of operations available in the CKKS scheme (see Section \ref{subsubsec:background:ckks:ops}), a neural network $f$ must be converted into a sequence of addition, multiplication, rotations, and rescaling operations. More precisely, consider a standard feed-forward neural network:

\begin{definition}[Feed-Forward Neural Network]
\label{def:neural_net}
    A \emph{neural network} computes a function $f:\inputSpace \to \mathbb{R}^O$ that maps vectors from the $I$-dimensional input space $\inputSpace \subseteq \mathbb{R}^I$ to an $O$-dimensional output space. A feed-forward architecture is a composition of layers $l\in[1,L]$ defined as
    %\begin{equation}
    %\label{eq:neural_net}
    $
    \mathbf{z}^{(l)} = \sigma^{(l)}\mleft(\tilde{\mathbf{z}}^{(l)}\mright) = \sigma^{(l)}\mleft(W^{(l)} \mathbf{z}^{(l-1)} + \mathbf{b}^{(l)}\mright),
    $
    %\end{equation}
    where $\sigma^{(l)}$ is an activation function, $\inputVal=\mathbf{z}^{(0)}\in\inputSpace$ is the network input, and $\mathbf{z}^{(L)}\in\mathbb{R}^O$ the network output. We denote by $f^{(l)}\mleft(\inputVal\mright)$ the value of $\tilde{\mathbf{z}}^{(l)}$ at layer $l$ for input $\inputVal$.
\end{definition}

While the linear term $W^{(l)} \mathbf{z}^{(l-1)} + \mathbf{b}^{(l)}$ can be natively implemented with the operations listed above, the activations $\sigma^{(l)}$ must be replaced with polynomial functions as follows:

%\footnote{For the sake of simplicity, we will often treat $\sigma$ as a non-linear element-wise transformation, even though its definition covers other operations such as average pooling and softmax.}

\begin{definition}[Polynomial Neural Network]
\label{def:poly_net}
    A \emph{polynomial neural network} computes a function $f_\pi:\inputSpace \to \mathbb{R}^O$ obtained 
    from a neural network $f:\inputSpace \to \mathbb{R}^O$ by replacing the activation functions $\sigma^{(l)}$ in Definition \ref{def:neural_net} by their polynomial approximation counterparts $\pi^{(l)}$ (see Sec.~\ref{ssec:range-analysis} for their construction). To differentiate the two networks,
    we rename the pre- and post-activation vectors from $\vec{z}$ to $\vec{y}$, i.e.\
    %\begin{equation}
    %\label{eq:neural_net-approx}
     $   \mathbf{y}^{(l)} = \pi^{(l)}\mleft(\tilde{\mathbf{y}}^{(l)}\mright) = \pi^{(l)}\mleft(W^{(l)} \mathbf{y}^{(l-1)} + \mathbf{b}^{(l)}\mright)$
    %\end{equation}
    with $\mathbf{y}^{(0)}=\mathbf{z}^{(0)}$ and $\mathbf{y}^{(l)}\approx\mathbf{z}^{(l)}$ for all layers $l>0$. We call $f_\pi^{(l)}\mleft(\inputVal\mright)$ the value of $\tilde{\mathbf{y}}^{(l)}$ for input $\inputVal$.
\end{definition}

As per Equation \eqref{eq:error_sources}, we call ${\bm e}_{\pi}(\mathbf{x})=f_{\pi}(\mathbf{x})-f(\mathbf{x})$ the approximation error introduced by polynomial networks. The common strategy to reduce this source of error is the use of higher-degree polynomials, at the cost of higher multiplicative depth in the CKKS circuit and thus an increase in latency \cite{Lee2022minimax}.

\subsection{Range Estimation}
\label{sec:sampled_ranges}

%\todo[inline]{
%Is ``effective polynomial networks'' the main reason? Shouldn't it be
%guaranteed output ranges? Or is this meant by ``effective''. We should
%be more clear here.
%}

\looseness=-1
To construct accurate polynomial approximations of the activation functions $\sigma^{(l)}$,
we must bound the range of their input, i.e.\ their pre-activation values $f_\pi^{(l)}\left(\vec{x}\right)$.
%of the polynomials' inputs. \alberto{It is not clear what polynomials' inputs are. Inputs are vectors, and we approximate activation functions, right?}
Therefore, we need to identify intervals $[\vec{l}^{(l)},\vec{u}^{(l)}]$ for each layer $l$ such that $f_\pi^{(l)}(\mathbf{x})\in[\vec{l}^{(l)},\vec{u}^{(l)}]$ is satisfied
for all inputs $\mathbf{x} \in \mathcal{I}$.
%\todo{$f_\pi^{(l)}(x)$ or $f^{(l)}(x)$?} f_\pi
%
If there exist inputs $\mathbf{x}$ where $f_\pi^{(l)}(\mathbf{x})$ is not in the computed approximation interval,
the approximation may result in large errors and uncontrollable growth, eventually leading to arithmetic failures.
This phenomenon 
is visible in the $\mathrm{GELU}$ approximations in Fig. \ref{subfig:gelu:approx} which are precise for the range $\left[-3,3\right]$, but rapidly diverge outside of it.

% If the true pre-activation values are out of range, the polynomial
% approximation may result in large errors and uncontrollable growth,
% eventually leading to arithmetic overflows in the CKKS circuit's
% computations. \alberto{I would explain this part better, as it is fundamental to this work. Furthermore, it is unclear how the content of this paragraph relates to the example of the GELU function approximation discussed in the next line.}

% %To build effective polynomial networks, we need to identify the expected ranges $\tilde{\mathbf{y}}^{(l)}\in\left[\mathbf{l}^{(l)},\mathbf{u}^{(l)}\right]$ of the polynomial inputs $\tilde{\mathbf{y}}^{(l)}=W^{(l)} \mathbf{y}^{(l-1)} + \mathbf{b}^{(l)}$ at each layer $l$. Given the ranges across the network, we can construct accurate polynomial approximations of the activation functions. 
% An example of a polynomial approximation of $\gelu{}$ on $\left[-3,3\right]$ can be seen in Fig.~\ref{subfig:gelu:approx}.
% \todo{Should we rescale the y-axis in Fig.~\ref{subfig:gelu:approx}}

\begin{figure}
    \vspace*{-1.2em}
    \begin{center}
    \includegraphics[width=0.8\linewidth]{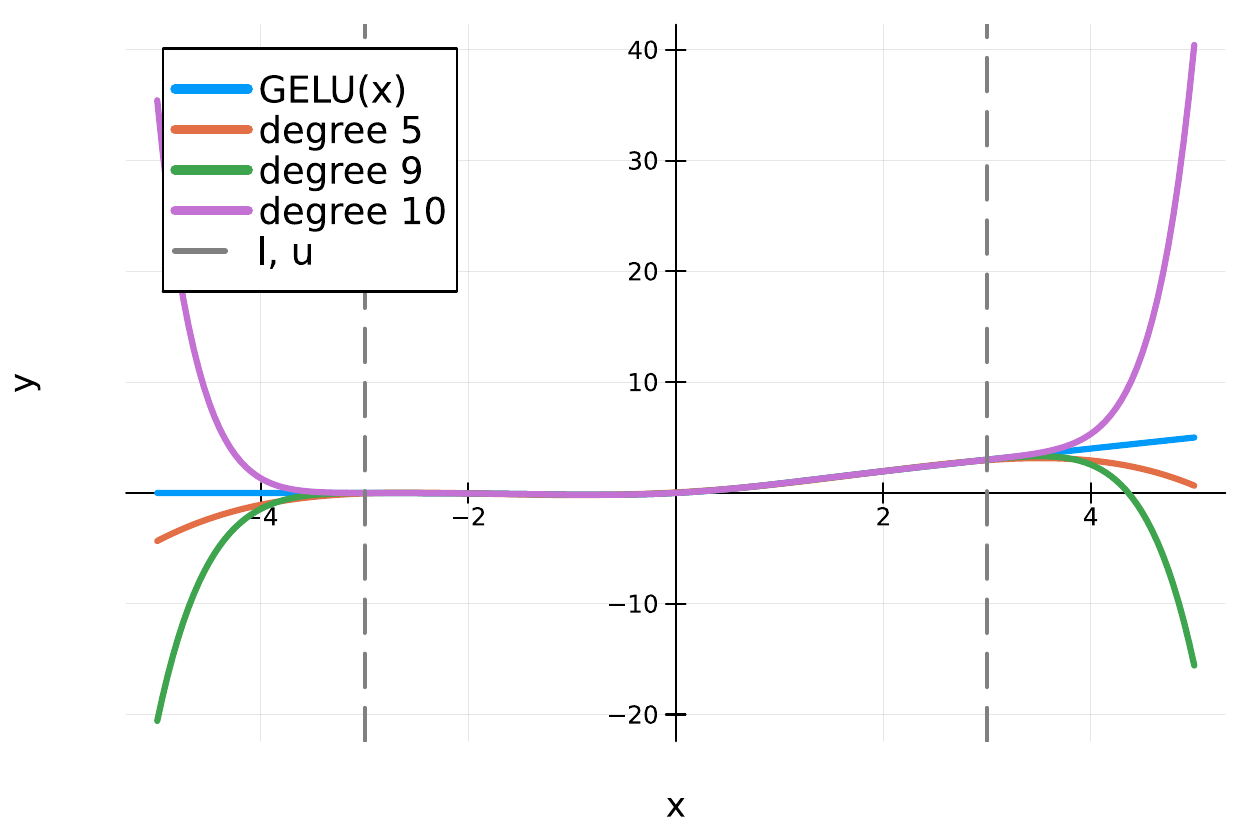}
    \end{center}
    \vspace*{-1.3em}
    \caption{Polynomial approximations diverge to $\pm \infty$ outside of the approximation interval.}
    \label{subfig:gelu:approx}
\end{figure}
\subsubsection{Sampling-Based Design} The most common approach to compute ranges is based on observations from a dataset $\dataset$ sampled from the training data distribution. For intermediate layers $l$, bounds $\mathbf{l}^{(l)},\mathbf{u}^{(l)}$ are computed such that $\mathbf{l}^{(l)} \leq f_\pi^{(l)}\mleft(\inputVal\mright) \leq \mathbf{u}^{(l)}$
%\todo{$f_\pi^{(l)}(x)$ or $f^{(l)}(x)$?} in my understanding typically f, not f_\pi
holds for all $\inputVal\in\dataset$.\footnote{For ease of implementation, the activation ranges are often estimated on the regular network $f^{(l)}\mleft(\inputVal\mright)$ rather than $f_\pi^{(l)}\mleft(\inputVal\mright)$.} 

Such sampling-based estimates provide no guarantees outside the dataset $\mathcal{D}$ and may significantly underestimate the true ranges.
%As a result, inputs that are out-of-distribution -- or perturbed -- might violate these ranges.
Existing countermeasures
%Countermeasures in the literature
include extending sampled
ranges by a factor of two~\cite{Ebel2025orion} or 
choosing a constant range of $[-c,c]$ for a suitable large constant $c>0$~\cite{Lee2022fhe}.
%choosing conservative ranges of $[-c,c]$ for the whole network \cite{Lee2022fhe} (for a suitable constant $c>0$).
%\todo{$[-40,40]$ sounds quite arbitrary; why not say $[-c,c]$ for a suitable constant $c$?}
Others propose to reduce the impact of outliers on ranges by pre-multiplying the weights by orthogonal matrices~\cite{LLAMACKKS}.
%Recently, Park \textit{et al.} \cite{LLAMACKKS} introduced a technique to reduce the impact of outliers on the range sizes by pre-multiplying the weights by orthogonal matrices.
%

In Section \ref{sec:overflow}, we show that such countermeasures do not protect against all range violations, thus leading to catastrophic failures in the CKKS network.

\subsubsection{Neural Network Verification}
\label{sec:background_nnv}
A principled alternative to the empirical estimation of ranges is the use of neural network verification techniques.
For classical (non-polynomial) neural networks, there are numerous tools~\cite{Xu21aCROWN,katz_marabou_2019,henriksen2020efficient} that can solve constraint and optimization problems on the computational graph of a neural network.
For example, such tools can be used on non-polynomial neural networks to compute pre-activation ranges by solving the following optimization problem for all $l\in\left[1,L\right]$:
\begin{align}
    &\min / \max~~ \tilde{\vec{z}}^{(l)}
    \label{eq:nn_optimization}
    \\
    &\text{s.t. }
    \vec{z}^{(0)}\in\inputSpace \text{ and }
    \left\{ 
        \begin{aligned}
            \tilde{\vec{z}}^{(j)} &= W^{(j)} \vec{z}^{(j-1)} + \vec{b}^{(j)}
            \\
            \vec{z}^{(j)} &= \sigma^{(j)}\left( \tilde{\vec{z}}^{(j)}\right)
        \end{aligned}
        \right.
        &&\text{for }j \leq l\;.
        \nonumber
\end{align}
The progress in NN verification scalability is recorded in the yearly competition~\cite{vnncomp2025}.
While this research field has seen significant advances in general,
many state-of-the-art techniques are not readily applicable to (high-degree) polynomial NNs.

\subsubsection{Differential Verification}
\label{sssec:diff-verification}
Standard NN verifiers are also notoriously bad at the verification of bounds on the \emph{difference} between two NN executions~\cite{Paulsen2020Reludiff,Paulsen2020Neurodiff,Teuber2025verydiff}, i.e.\ ${f_1\left(\vec{x}\right)-f_2\left(\vec{x}\right)}$ for distinct NNs $f_1$ and $f_2$.
Hence, specialized solvers~\cite{Paulsen2020Reludiff,Paulsen2020Neurodiff,Teuber2025verydiff} propose \emph{differential verification} propagating bounds on the execution difference \emph{up to layer} $l$ (i.e.\ ${f_1^{(l)}\left(\vec{x}\right)-f_2^{(l)}\left(\vec{x}\right)}$) through the NN.
For example, the difference ${\tilde{\vec{\Delta}}^{(l)} = \tilde{\vec{y}}_1^{(l)} - \tilde{\vec{y}}_2^{(l)}}$ between two pre-activation values can be expressed in terms of the difference ${\vec{\Delta}^{(l-1)} = \vec{y}_1^{(l-1)} - \vec{y}_2^{(l-1)}}$
between the results of the previous layer via the identity
%between the post activation values of the previous layer via the identity
\[
    \tilde{\vec{\Delta}}^{(l)} = \left(W_1^{(l)} - W^{(l)}_2\right) \vec{y}^{(l-1)}_2 + W^{(l)}_1 \vec{\Delta}^{(l-1)} + \left(\vec{b}^{(l)}_1 - \vec{b}^{(l)}_2\right).
\]
The post-activation difference can be similarly rewritten as 
\[
\vec{\Delta}^{(l)} = \sigma^{(l)}_1(\tilde{\vec{y}}_1^{(l)}) - \sigma^{(l)}_2(\tilde{\vec{y}}_2^{(l)})=\sigma^{(l)}_1(\tilde{\vec{y}}_1^{(l)}) - \sigma^{(l)}_2(\tilde{\vec{y}}_1^{(l)} - \tilde{\vec{\Delta}}^{(l)} )\;.
\]
While this allows to efficiently bound the difference between executions of classical NNs, this approach~\cite{Paulsen2020Neurodiff,Teuber2025verydiff} is also not directly applicable to high-degree polynomial NNs.

%$f^{(l)}\mleft(\inputVal\mright)$ rather than $f_\pi^{(l)}\mleft(\inputVal\mright)$

To address both of these limitations, a recent line of work leverages NN verification for sound range estimation and to bound the numerical error between a regular NN $f\mleft(\inputVal\mright)$ and its polynomial variant $f_\pi\mleft(\inputVal\mright)$~\cite{Manino2023fhe,Kern2025fhe}.
LiGAR~\cite{Manino2023fhe} allocates the optimal polynomial approximations for an NN under a given maximal error budget by extending the  optimization problem in Equation~\eqref{eq:nn_optimization} with bounded error variables by redefining $\vec{z}^{(j)}$ as follows:
\begin{align}
\label{eq:nn_delta_optimization}
\vec{z}^{(j)} &= \sigma^{(j)}\left( \tilde{\vec{z}}^{(j)}\right) + \vec{\delta}^{(j)}\;,
\end{align}
where the admissible range of $\vec{\delta}^{(j)}$ is determined via an iterative refinement-cycle.
Importantly, this approach only works with the degree of the best polynomial approximation, but does not actually compute the polynomial coefficients.
Furthermore, it estimates the impact via (possibly conservative) Lipschitz constants.

In contrast, \textsc{ZonoPoly}~\cite{Kern2025fhe} does compute concrete polynomial coefficients and can prove tighter bounds than \textsc{LiGAR}. However, its verification strategy is based on zonotope propagation, which cannot compete with state-of-the-art NN verifiers with respect to computing tight pre-activation ranges. Additionally, both \textsc{LiGAR} and \textsc{ZonoPoly} are only applicable to neural networks $f\mleft(\inputVal\mright)$ with $\relu$ activations, and do not demonstrate whether their bounds translate to a robust CKKS circuit. These are severe limitations, as smoother activations can drastically reduce the approximation error (see Section~\ref{sec:poly-approx-smooth}) and designing a CKKS circuit that can accommodate robust polynomial approximations requires non-trivial work (see Section~\ref{sec:concrete-ckks}).

\section{Overflow Events in CKKS Networks}
\label{sec:overflow}

At an abstract level, private inference over encrypted data consists of computing the \textit{total} function $f_\text{CKKS}(\mathbf{x})=f_\pi(\mathbf{x}) + {\bm e}_\text{CKKS}(\mathbf{x})$, which will always return \textit{some} value. Any issue that may occur during the computation is by definition undetectable until the output is decrypted.

One such issue is that of a \textit{plaintext overflow}~\cite{10.1145/3626232.3653277}. More specifically, consider that the plaintext space of CKKS, i.e. $\mathcal{R}$ %\todo{add the modulus here?} 
is embedded in $\mathcal{R}_{Q_\ell}^2$ for security reasons. This implies that, any coefficient in the plaintext polynomial being larger than $Q_\ell/2$ can corrupt the encoded message. In order to avoid an explosion in the small ciphertext error ${\bm e}_\text{CKKS}(\mathbf{x})$, the infinity norm of any plaintext message must remain smaller than half the modulus $Q_\ell$ in which it is encrypted. Whenever such condition is violated, the plaintext coefficients will wrap around $Q_\ell$, thus corrupting the result:

\begin{definition}[Plaintext Overflow]
\label{def:plaintext_overflow}
    For any ciphertext ${\bm c}$ encrypting ${\bm m} \in \mathcal{R}$, we say that ${\bm c}$ contains a \emph{plaintext overflow} if the underlying ${\bm m}$ is such that $\|{\bm m}\|_\infty \ge Q_\ell/2$, where $Q_\ell$ is the RLWE modulus.
\end{definition}

Plaintext overflows happen when the message to be represented grows too large. Here, we argue that evaluating high-degree polynomial activations may easily trigger such an event. Indeed, a single activation exceeding its approximation interval may suffice to yield values that quickly diverge to $\pm \infty$ (see Figure \ref{subfig:gelu:approx}), thus rapidly exhausting the available modulus $Q_\ell$. We formalise an overflow event as follows:

\begin{definition}[Polynomial Overflow Events]
\label{def:overflow_atk}
\looseness=-1
Let $f_\pi:\inputSpace \to \mathbb{R}^O$
be a polynomial network with expected ranges
$[\mathbf{l}^{(l)},\mathbf{u}^{(l)}]$ for all
layers $l$.
An input $\vec{x} \in \mathcal{I}$ will cause an \emph{overflow event on $f_\pi$} iff there is a layer $l$ where
$f^{(l)}_\pi\left(\vec{x}\right)\not\in [\mathbf{l}^{(l)},\mathbf{u}^{(l)}]$. 
\end{definition}

\looseness=-1
Note that there is no one-to-one correspondence between plaintext and polynomial overflows: the modulus $Q_\ell$ might be large enough to accommodate intermediate layers $f^{(l)}_\pi\left(\vec{x}\right)$ that exceed their expected ranges $[\mathbf{l}^{(l)},\mathbf{u}^{(l)}]$ by a very small amount. However, more significant violations can exhaust the modulus and cause unrecoverable outputs. Indeed, we show how such inputs can easily be constructed in Section \ref{sec:overflow_search} and demonstrate how this issue can occur in practice in Section \ref{sec:design_atk_compare}.

In contrast, we would like to design polynomial networks for which polynomial overflow events can never occur. That way, the plaintext message will always remain within its design constraints from Definition \ref{def:plaintext_overflow}:

\begin{definition}[Overflow Robustness]
\label{def:overflow_rob}
\looseness=-1
    A network $f_\pi$ is \emph{overflow robust} on $\mathcal{I}$ if no $\vec{x} \in \mathcal{I}$ can cause an overflow event on $f_\pi$.
\end{definition}

\noindent We give one such construction in Section \ref{sec:cert_design}.

\subsection{A Note on the Threat Model}

Before introducing our method to find overflow events, let us clarify our assumptions on the threat model. Here, we are concerned with the \textit{possibility} of overflow events. That is, an honest user providing a legitimate input which \textit{accidentally} causes an overflow event. %In Table \ref{tab:eval:overflow-attack}, the reader can observe that overflow can also naturally happen simply by replacing the test set (e.g., the MNIST case).
% E.g. by using as input the USPS dataset over a MNIST trained network, we observed that 5% of the inputs caused overflows)
This can already happen when processing images from a different test set: for instance, training a handwritten digits classifier on the MNIST dataset and evaluating it on USPS  (see Table~\ref{tab:eval:overflow-attack}).
% We remark that our method to find such overflows implies having access to the weights of the network, which far from real-world scenarios. However, we want to stress that the point of finding overflows is not about causing overflow itself, but rather to show that such overflows can occur in realistic inputs (i.e., not  
% 

Below, we additionally provide an auditing technique that searches for benign-looking inputs that cause overflows. While our technique constructs inputs in an adversarial fashion, these inputs could also be supplied in practice by an honest user.
%they are meant to represent legitimate data that an honest user might want to process.

% while executing our method requires access to the networks's weights, it still demonstrates that overflow events can happen even for realistic inputs that may occur in a non-adversarial setting in practice.

\subsection{Finding Overflows}
\label{sec:overflow_search}
%

% regarding Lorenzo's comment for the gamma function:
% it is a combination of
% - adding noise
% - rotating the image
% - small x/y translation of the image
% Thanks :-)))

To uncover overflow events in polynomial networks, we need to find legitimate inputs $\vec{x} \in \mathcal{I}$ that exceed the expected ranges $[\mathbf{l}^{(l)},\mathbf{u}^{(l)}]$. To do so, we leverage classic techniques from adversarial example generation: for any given benign input $\inputVal\in\inputSpace$, we construct a perturbed input $\inputVal'=\perturbedVersion\mleft(\inputVal,\noiseVector\mright)$. Here, the parameter $\noiseVector$ regulates the intensity and direction of local perturbations around $\inputVal$. In our experiments (see Section \ref{sec:design_atk_compare}), $\perturbedVersion$ is a combination of small image rotations, small translations, and additive noise (see, e.g., Figure \ref{fig:mnist-perturbations}).

%For a given input $\inputVal\in\inputSpace$ we generate random noise $\noiseVector$ which we use to produce a perturbed version of the input (denoted $\perturbedVersion\mleft(\inputVal,\noiseVector\mright)$).\todo{LR: if we have space, one line to better define such $\perturbedVersion$?}

We then use projected gradient descent~\cite{DBLP:conf/iclr/MadryMSTV18} to optimize $\noiseVector$ towards inputs $\inputVal'$ that violate the expected ranges of $f_{\pi}$. Since direct optimization on $f_{\pi}$ is numerically unstable, we optimize with respect to $f$ instead. Note that the internal activation values $f^{(l)}\left(\vec{x}\right)$ and $f_\pi^{(l)}\left(\vec{x}\right)$ are tightly related (until the latter violates its bound) and thus can act as a natural surrogate of each other.
More formally, we maximize the distance between the internal activation values and the centre point of their expected ranges $(\mathbf{u}^{(l)}+\mathbf{l}^{(l)})/2$:
%Formally, to attack a network at layer $l$, we minimize $\noiseVector$ w.r.t. the loss function
\begin{equation}
\label{eq:adv_loss}
\sum_j
\left\lVert \left(f^{(l)}\mleft(\perturbedVersion\mleft(\inputVal,\noiseVector\mright)\mright)\right)_j-\left(\mathbf{u}^{(l)}_j+\mathbf{l}^{(l)}_j\right)/2\right\rVert_\infty.
\end{equation}%
%Overflows occur once the distance from the midpoint of layer $i$ exceed $(\mathbf{u}^{(l)}-\mathbf{l}^{(l)})/2$ which makes our loss a natural surrogate for boundary-crossing.
%
As will be shown below, optimising over small perturbations $\noiseVector$ of benign inputs $\inputVal$ generates realistic inputs causing overflows.

\begin{table*}[t]
    \centering
    \caption{Comparison between sampled ranges and certified construction of CKKS networks. While the design-time accuracy is similar, the former are susceptible to overflows. These already happen naturally on MNIST, as soon as we change the test set, but our auditing technique (Found Overflows) reveals a larger percentage of them on both HELOC and CIFAR10.}
    \footnotesize
    \setlength{\tabcolsep}{5.2pt} % Default value: 6pt
    \renewcommand{\arraystretch}{1.25} % Default value: 1
    \begin{tabular}{l l  l r r l r r l r r}
        & & \multicolumn{3}{c}{\makecell{\textbf{Design-Time Accuracy}\\(higher is better)}} & \multicolumn{3}{c}{\makecell{\textbf{Naturally-Occurring Overflows}\\(lower is better)}} & \multicolumn{3}{c}{\makecell{\textbf{Found Overflows Events}\\(lower is better)}}
        \\[1.5mm]\cline{3-11}
        \textbf{Network} & \textbf{Size} & \textbf{Dataset} & \textbf{Sampled} & \textbf{Certified} & \textbf{Dataset} & \textbf{Sampled} & \textbf{Certified} & \textbf{Dataset} & \textbf{Sampled} & \textbf{Certified} \\
        \midrule
        \textbf{HELOC} & $[64,32]$
        & HELOC (1) & \textbf{71.86\%} & \textbf{71.86\%}
        & HELOC (2) & \textbf{0.0\%} & \textbf{0.0\%}
        & HELOC (2) & 0.05\% & \textbf{0.0\%}\\
        \midrule
        \textbf{MNIST} & [256,256,256,256]
        & MNIST & \textbf{97.80\%} & \textbf{97.80\%}
        & USPS & 5.18\% & \textbf{0.0\%}
        & USPS & 12.30\% & \textbf{0.0\%}\\
        \midrule
        \textbf{CIFAR10} & $[900,392,144,256]$
        & CIFAR 10 & \textbf{69.54\%} & 69.50\%
        & CIFAR 10.1 & \textbf{0.0\%} & \textbf{0.0\%}
        & CIFAR 10.1 & \hphantom{0}0.10\% & \textbf{0.0\%}\\
        & $[1.8\mathrm{k},784,288,256]$
        & CIFAR 10 & \textbf{77.28\%} & 75.13\%
        & CIFAR 10.1 & \textbf{0.0\%} & \textbf{0.0\%}
        & CIFAR 10.1 & 46.70\% & \textbf{0.0\%}\\
    \end{tabular}
    \label{tab:eval:overflow-attack}
\end{table*}

%\todo{I've moved the section here now, note we have to talk about certified design now before introducing it, but I think that's OK?}

\begin{figure}[b]
\centering
\begin{subfigure}{0.49\linewidth}
        \includegraphics[width=\linewidth]{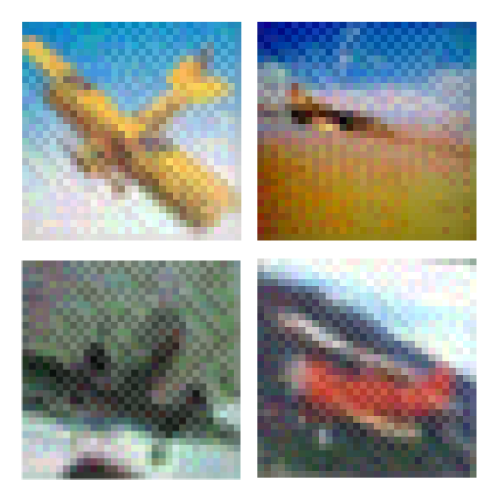}
\end{subfigure}
\hfill
\begin{subfigure}{0.49\linewidth}
        \includegraphics[width=\linewidth]{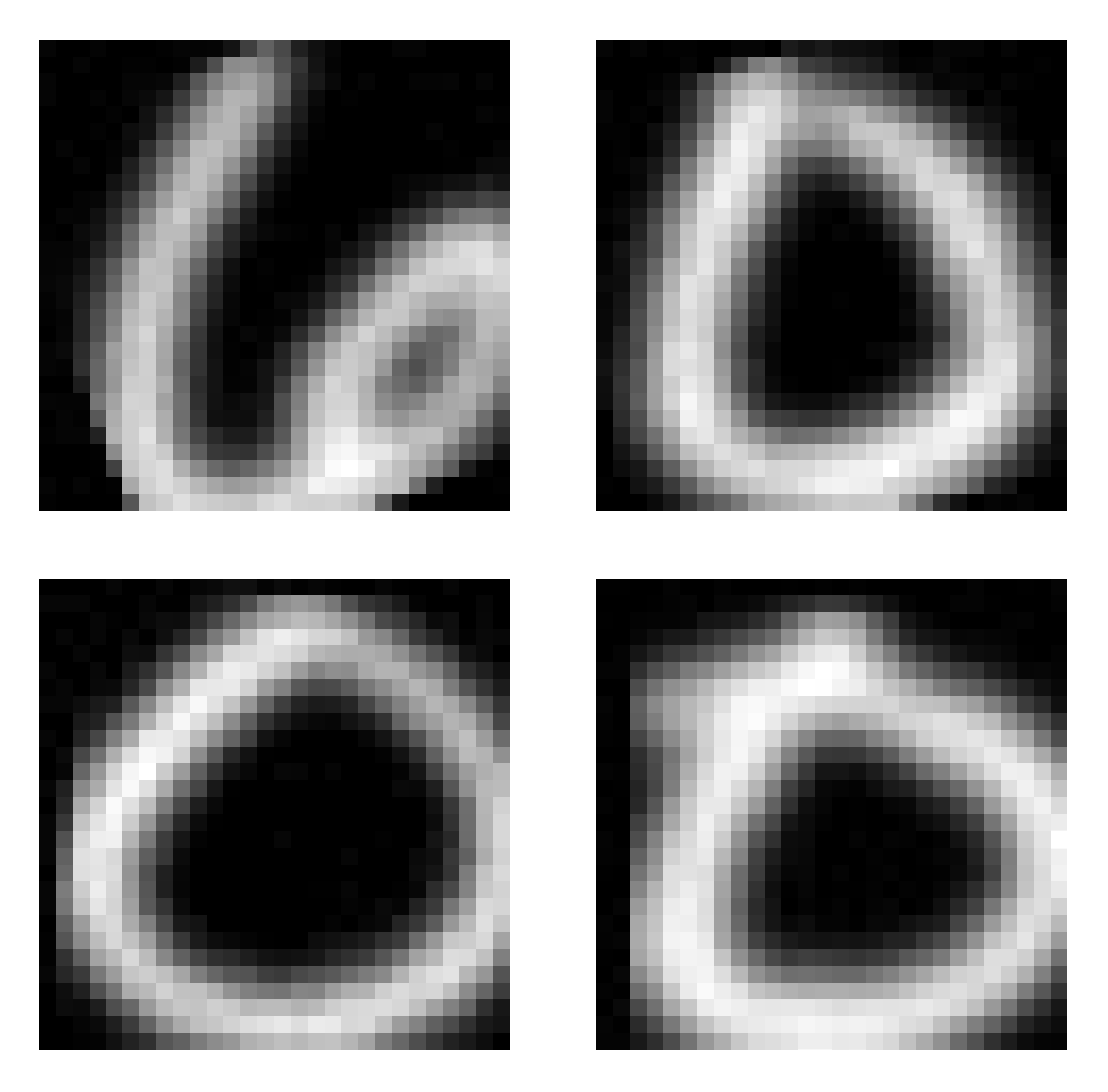}
\end{subfigure}
        \caption{Realistic perturbations of CIFAR and MNIST images causing overflow events}
        \label{fig:mnist-perturbations}
\end{figure}

\subsection{Overflow Events in Sampling-Based Networks}
\label{sec:design_atk_compare}

To demonstrate the vulnerability of sampling-based CKKS network construction to overflow events, we deploy the algorithm described above on classification NNs for tabular data (HELOC~\cite{FICOChallenge}), handwritten digit classification (MNIST~\cite{Lee2022fhe}), and image classification (CIFAR~\cite{krizhevsky2009cifar}).
We compare the vulnerability of a CKKS network constructed via sampled bounds to the vulnerability of an NN constructed with our certified approach described in Section \ref{sec:cert_design}.
In both cases, we use $\mathrm{GeLU}$ activation functions and heterogeneous polynomials.
Table \ref{tab:eval:overflow-attack} provides an overview and demonstrates that our approach presented in Section \ref{sec:cert_design} yields overflow robust networks.
Notably, the overflow-event rate yields an \emph{upper bound} on the possible accuracy of the CKKS circuit:
If $p\; \%$ of samples have an overflow event, then the accuracy of the circuit can be at most $100-p\;\%$ (as the remaining predictions cannot be decoded).

\looseness=-1
For HELOC we generate sampled bounds using the first half of the dataset and use the other half for finding overflows.
We allow perturbations of $15\%$ w.r.t. each feature's range.
To discretize values, our adversarial procedure uses a straight-through-estimator~\cite{DBLP:journals/corr/BengioLC13}.
We find three overflows on the sampled NN.

%\paragraph{MNIST.}
For MNIST we generate sampled bounds using the original MNIST dataset, then we search for overflow-inducing inputs on the USPS dataset~\cite{DBLP:journals/pami/Hull94} (same input configuration and task setting). Note that the USPS dataset already causes overflows in $5.18\%$ of cases without the need of adding any input perturbations.
When we allow rotations of $23^\circ$, translation by 20\% and $L_\infty$-noise of up to $\varepsilon=5/255$, we manage to increase the overflow rate to approx. 12.3\% of cases, while still producing realistic-looking inputs.

%\paragraph{CIFAR.}
For CIFAR, we generate sampled bounds using the original CIFAR dataset.
We search for overflows on the independent CIFAR10.1 dataset~\cite{torralba2008tinyimages,recht2018cifar10.1} (same input configuration and task setting).
We allow rotations of $17^\circ$ and $L_\infty$ noise of up to $\varepsilon=15/255$ (no translation due to non-black background in images).
We find overflows in both a small and a deeper CIFAR NN.
For the latter, we manage to cause an overflow on $46.7\%$ of images (see Figure \ref{fig:mnist-perturbations}).
This implies the CKKS circuit accuracy is below 53.3\% on the evaluated dataset.

\paragraph*{Reproduction in CKKS Circuit and Discussion}
Our adversarial procedure initially only searches for perturbations that push the polynomial NN $f_\pi$ outside its bounds.
As discussed above, this overflow is likely to induce decoding errors for the CKKS circuit.
We validated this assumption on a subset of MNIST inputs and showed that all overflow events on $f_\pi$ lead to arithmetic errors in the CKKS circuit, which underscores the importance of eliminating overflows.
Table \ref{tab:eval:overflow-attack} shows that our certified design introduced in the next section eliminates overflows without major impacts on accuracy.
The subsequent sections will describe a principled approach for the construction of such certified CKKS circuits.

%\todo{Move 7.3 (including Figure 8 and Table 3) to here}

\section{Avoiding Overflows by Certified Design}
\label{sec:cert_design}

Given a neural network $f$, we propose a method to construct robust CKKS circuits in the sense of Definition \ref{def:overflow_rob}, i.e.\ NNs $f_{\pi}$ without overflows,
while simultaneously bounding the output difference $\| f_{\pi}(\vec{x}) - f(\vec{x}) \|_\infty$.
%. Our approach furthermore provides provable bounds on the output difference $\| f_{\pi}(\vec{x}) - f(\vec{x}) \|_\infty$.
Our approach transforms $f$ into its polynomial version $f_{\pi}$ by iteratively rewriting its layers.
For each layer $l$, we first compute provable bounds such that ${f^{(l)}_\pi\left(\vec{x}\right)\in[\vec{l}^{(l)},\vec{u}^{(l)}]}$ for all $\vec{x}\in\mathcal{I}$ (Section \ref{ssec:range-analysis}).
Linear layers then remain unchanged while activation functions are substituted by a polynomial approximation, obtained on the fly for the interval $[\vec{l}^{(l)},\vec{u}^{(l)}]$ using the Remez algorithm~\cite{remez1934}.
The algorithm relies on a bound for the maximum (local) error between the activation function and its approximation (Section \ref{ssec:pw-overapprox}) which is taken into account for the range analysis in subsequent layers.
We additionally leverage differential verification via zonotopes~\cite{Teuber2025verydiff} to compute bounds on the difference $\|f_{\pi}^{(l)}(\vec{x})-f^{(l)}(\vec{x})\|_\infty$ (Section \ref{ssec:diff-verification}).
Furthermore, we provide a direct translation from the generated polynomial network $f_{\pi}$ to a CKKS circuit (Section \ref{sec:cert-nn-to-ckks}).

\subsection{Sound Range Analysis for Polynomial Networks}
\label{ssec:range-analysis}
We construct the polynomial NN $f_{\pi}$ from a given NN $f$ by iteratively substituting a non-linear activation function $\sigma^{(l)}$ in layer $l$ by its polynomial approximation $\pi^{(l)}$.
The latter is computed by the Remez algorithm~\cite{remez1934} which requires a method for computing a sound bound on $| \pi^{(l)}(\mathbf{x}) - \sigma^{(l)}(\mathbf{x}) |$ (see Section~\ref{ssec:pw-overapprox}) and an approximation range $[\vec{l}^{(l)},\vec{u}^{(l)}]$.
%This section presents how we obtain a sound approximation interval from our input space $\mathcal{I}$.
% To construct the polynomial network $f_{\pi}$ from a given neural network $f$,
% we iteratively compute interval bounds $[\vec{l}^{(l-1)},\vec{u}^{(l-1)}]$ w.r.t. the input space $\mathcal{I}$.
% We then substitute a non-linear activation function $\sigma^{(l)}$ in layer $l$ by its polynomial approximation $\pi^{(l)}$ on the range $[\vec{l}^{(l-1)},\vec{u}^{(l-1)}]$.
% We compute the latter using the Remez algorithm~\cite{remez1934}.

%By relying on \emph{verified} bounds for the reachable values of layer $\left(l-1\right)$, we ensure the absence of catastrophic overflows.
%

To compute valid ranges for $f_{\pi}^{(l)}$ (beyond the first polynomial layer) it is insufficient to analyse NN $f$ in isolation, as its output bounds also depend on the polynomial approximation errors introduced in earlier layers.
To this end, we follow the approach by Manino \emph{et al.}~\cite{Manino2023fhe} and compute sound ranges by
optimizing the problem given in Equation (\ref{eq:nn_optimization}) with the modification given in Equation (\ref{eq:nn_delta_optimization}).
However, in contrast to prior work, we bound $\vec{\delta}^{(j)}$ a priori  with the interval $[-\vec{\epsilon}^{(l)},\vec{\epsilon}^{(l)}]$
where $\vec{\epsilon}^{(l)}$ is computed as seen in Section~\ref{ssec:pw-overapprox} (i.e., $| \pi^{(l)}(\vec{x}) - \sigma^{(l)}(\vec{x}) |\leq\vec{\epsilon}^{(l)}$).

\looseness=-1
As the resulting optimization problem only contains regular activations $\sigma^{(l)}$ (instead of polynomials $\pi^{(l)}$), % using only the original activation functions,
we can use state-of-the-art NN verifiers~\cite{Xu21aCROWN} to obtain guaranteed bounds.
This approach allows us to construct $f_{\pi}$ in a single pass over $f$ with the soundness guarantee given in Theorem \ref{thm:overall-soundness} (see proof on page \pageref{proof:prop:cert-construction}).
Below, we describe how errors on the $f_{\pi}$ are bounded locally (Section \ref{ssec:pw-overapprox}) and globally (Section \ref{ssec:diff-verification}).
%Thus, we construct $f_{\pi}$ in a single pass over $f$ while guaranteeing that the resulting network $f_{\pi}$ is robust to overflows :
\begin{theoremE}[Soundness][end, text link={}]
\label{thm:overall-soundness}
For any neural network $f$ with input space $\mathcal{I}$, the construction approach given in Section \ref{ssec:range-analysis} yields an overflow robust network $f_{\pi}$ on $\mathcal{I}$ (see Definition \ref{def:overflow_rob}).
\end{theoremE}
\begin{proofE}
\label{proof:prop:cert-construction}
A neural network $f_{\pi}$ is overflow robust iff for all layers $l$ it holds that its polynomials $\pi^{(l)}$ are designed for an interval $[\vec{l}^{(l)},\vec{u}^{(l)}]$ such that for all $\vec{x}\in\mathcal{I}$ it holds that $f^{(l)}_{\pi}\left(\vec{x}\right) \in [\vec{l}^{(l)},\vec{u}^{(l)}]$.

Note that our approach computes the intervals $[\vec{l}^{(l)},\vec{u}^{(l)}]$ by solving the maximization problem in equation (\ref{eq:nn_optimization}) with the modification given in (\ref{eq:nn_delta_optimization}).
To prove that our approach generates certified neural networks $f_{\pi}$ we must thus prove that at the time of polynomial construction the optimization problem in equations (\ref{eq:nn_optimization}) and (\ref{eq:nn_delta_optimization}) correctly over-approximates the behaviour of $f^{(l)}_{\pi}$.
Before we do so below, observe that this property suffices to guarantee certified network construction:
Every time an activation function is substituted by its polynomial approximation, we then compute valid bounds for $f^{(l)}_{\pi}$ using equations (\ref{eq:nn_optimization}) and (\ref{eq:nn_delta_optimization}) and construct an appropriate polynomial using the Remez algorithm for the sound range estimate.

Formally, proving that equations (\ref{eq:nn_optimization}) and (\ref{eq:nn_delta_optimization}) correctly over-approximates the behaviour of $f^{(l)}_{\pi}$ amounts to showing that for any $\vec{x}\in\mathcal{I}$ and any $1 \leq l \leq L$ it holds that there exist assignments to the optimization problem such that $\tilde{\vec{y}}^{(l)}=f^{(l)}_{\pi}\left(\vec{x}\right)$.

We prove the latter inductively:
First, observe that for $l=1$ the equations \emph{exactly} represent the behaviour of $f^{(1)}_{\pi}$.
As $f^{(1)}_{\pi}=f^{(1)}$, $\tilde{\vec{y}}^{(l)}$ and $f^{(l)}_{\pi}\left(\vec{x}\right)$ naturally coincide for all $\vec{x}\in\mathcal{I}$.

For our induction step ($l+1$), we assume that we have shown the over-approximation property for layer $l$.
Due to this, we know that the Remez algorithm computes a polynomial approximation $\pi^{(l)}$ for $\sigma^{(l)}$ on the sound range $[\vec{l}^{(l)},\vec{u}^{(l)}]$ for which we know that for all $\vec{x}\in\mathcal{I}$ we have $f^{(l)}\left(\vec{x}\right)\in[\vec{l}^{(l)},\vec{u}^{(l)}]$.
Since we assume correct local bound computation (see Section \ref{ssec:pw-overapprox} for detailed approach), we also know that $\left| \pi^{(l)}\left(f^{(l)}\left(\vec{x}\right)\right)-\sigma^{(l)}\left(f^{(l)}\left(\vec{x}\right)\right)\right|\leq\vec{\epsilon}^{(l)}$ for any $\vec{x}\in\mathcal{I}$.

For any $\vec{z}^{(0)}=\vec{x}\in\mathcal{I}$, we now consider some concrete assignment of the variables in equations (\ref{eq:nn_optimization}) and (\ref{eq:nn_delta_optimization}) such that $\tilde{\vec{z}}^{(l)} = f^{(l)}\left(\vec{x}\right)$ (exists due to induction hypothesis).
We can then set $\vec{\delta}^{(l)}$ to
$\pi^{(l)}\left(\tilde{\vec{z}}^{(l)}\right)-\sigma^{(l)}\left(\tilde{\vec{z}}^{(l)}\right) = \pi^{(l)}\left(f^{(l)}\left(\vec{x}\right)\right)-\sigma^{(l)}\left(f^{(l)}\left(\vec{x}\right)\right)$
which we know to satisfy the constraint $\vec{\delta}^{(l)}\in[-\vec{\epsilon}^{(l)},\vec{\epsilon}^{(l)}]$.
Consequently, we can assign $\vec{z}^{(l)}$ to $\pi^{(l)}\left(\tilde{\vec{z}}^{(l)}\right)$ (as $\sigma^{(l)}\left(\tilde{\vec{z}}^{(l)}\right)+\left(\pi^{(l)}\left(\tilde{\vec{z}}^{(l)}\right)-\sigma^{(l)}\left(\tilde{\vec{z}}^{(l)}\right)\right)=\pi^{(l)}\left(\tilde{\vec{z}}^{(l)}\right)$).
Finally, we assign $\tilde{\vec{z}}^{(l+1)}$ to $W^{(l+1)} \vec{z}^{(l)} + \vec{b}^{(l+1)}$.
First, note that both of these assignments satisfy all constraints of equations (\ref{eq:nn_optimization}) and (\ref{eq:nn_delta_optimization}).
Second, observe that by definition of $f_{\pi}$, if $\tilde{\vec{z}}^{(l)} = f^{(l)}\left(\vec{x}\right)$ then now it also holds that $\tilde{\vec{z}}^{(l+1)} = f^{(l+1)}\left(\vec{x}\right)$.
This concludes our induction step of the over-approximation property.
\end{proofE}

\subsection{Piecewise-Polynomial Overapproximation of Activation Functions}
\label{ssec:pw-overapprox}
\begin{figure*}[ht]
    \begin{minipage}[b]{0.32\linewidth}
        \centering 
        \includegraphics[width=\linewidth]{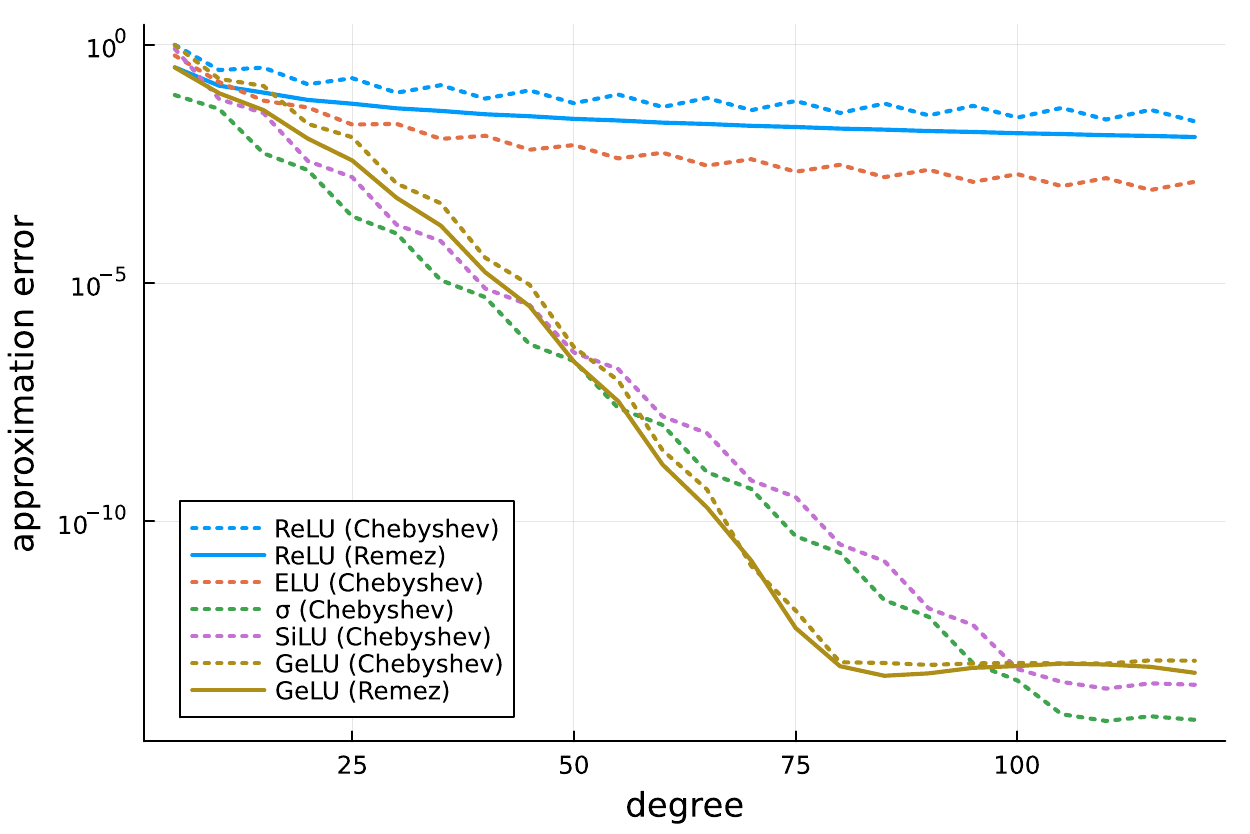}
        \captionof{figure}{Polynomial approximation error for functions over ${x \in [-1, 1]}$: Smooth functions converge faster.}
        \label{fig:convergence}
    \end{minipage}\hfill
    \begin{minipage}[b]{0.32\linewidth}
        \centering
        \includegraphics[width=0.9\linewidth]{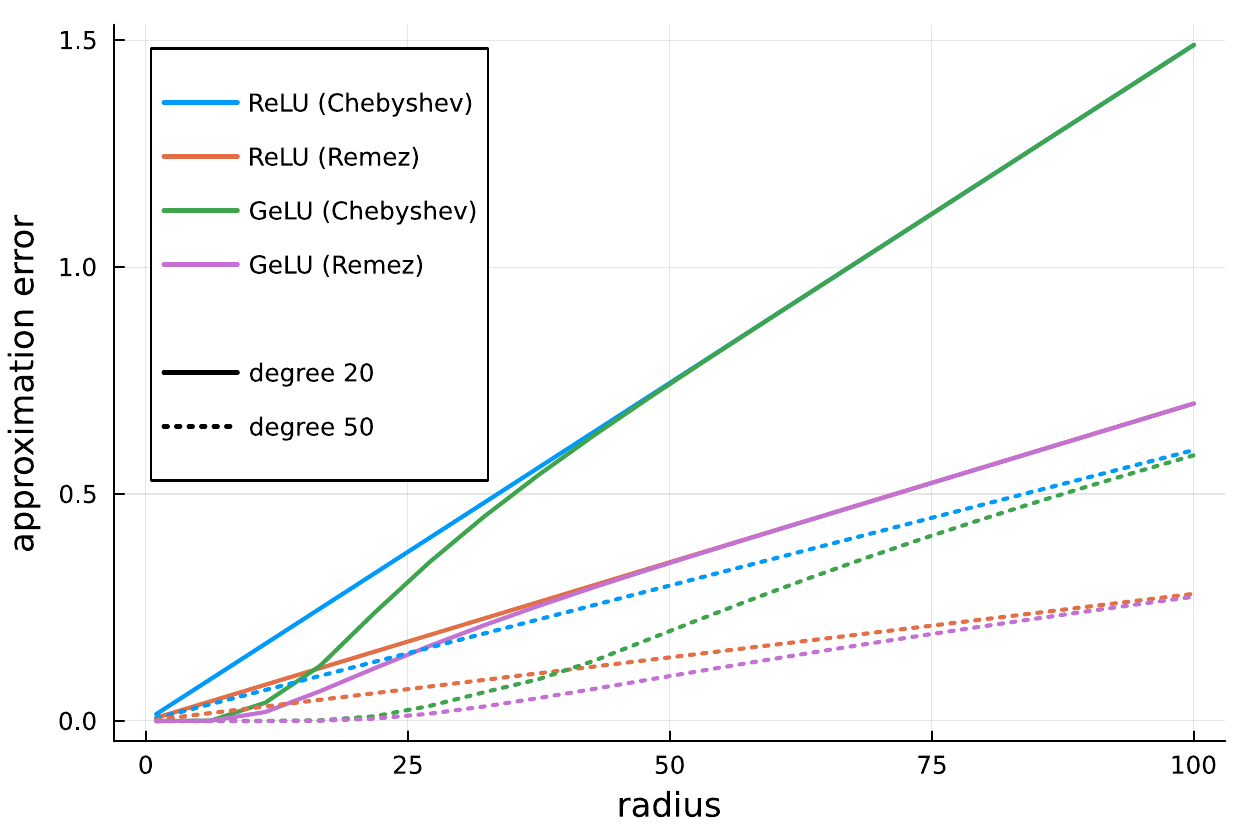}
        \captionof{figure}{Polynomial approximation error for $\mathrm{ReLU}$ and $\gelu$ for different intervals of size $[-r, r]$.}
        \label{fig:convergence-radius}
    \end{minipage}\hfill 
    \begin{minipage}[b]{0.32\linewidth}
        \centering
        \includegraphics[width=0.9\linewidth]{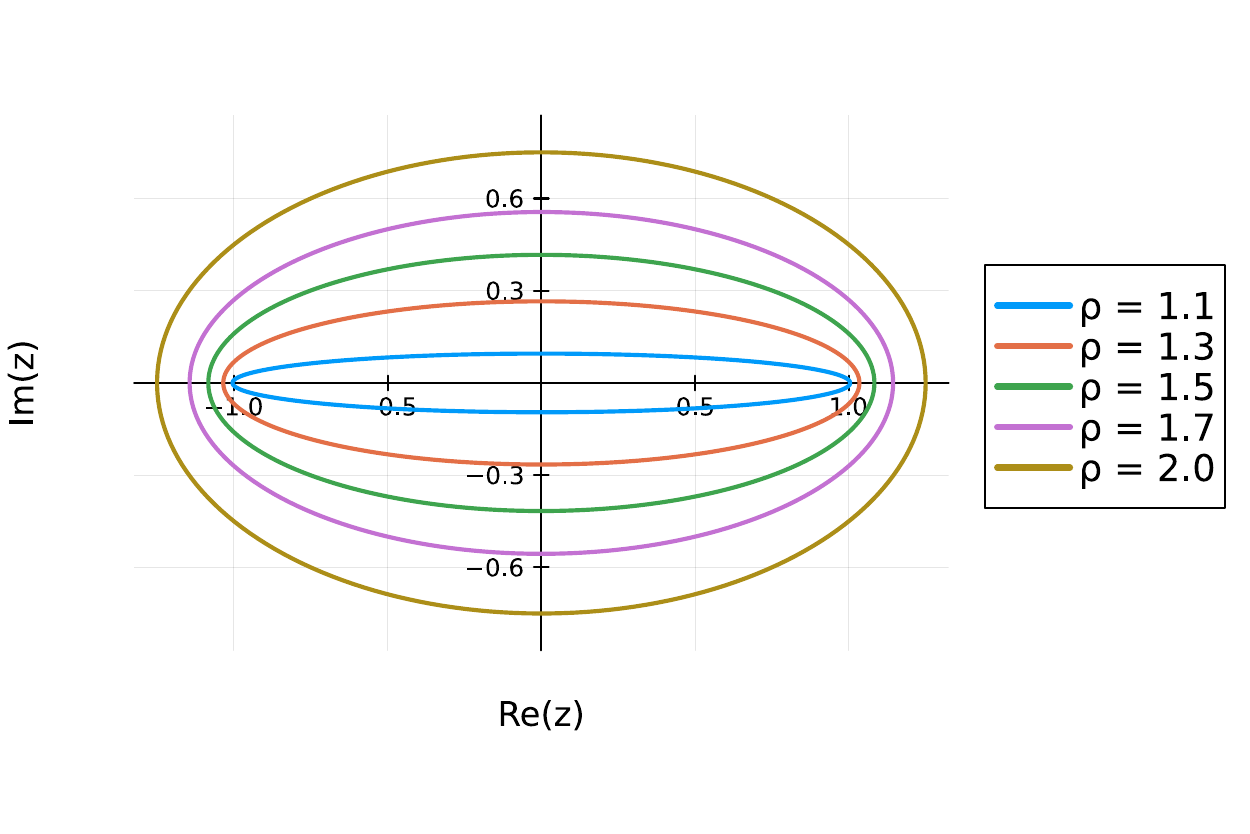}
        \captionof{figure}{Boundary of Bernstein ellipses $E_\rho$ for different values of parameter $\rho$.\\} %newline necessary to make all subcaptions 3 lines and align the images vertically
        \label{fig:bernstein-ellipse}
    \end{minipage}
    %\caption{Smoothness of activation functions affects the polynomial approximation error.}
\end{figure*}

%\todo[inline]{Sec Title: something more general about \textit{all} activations, not just GELU}
%\todo[inline]{Text: present ReLU briefly (cite literature). Explain why smooth functions are trickier, than use GELU as example, and claim novelty of our approach.}
Computing bounds on the univariate approximation error $\pi^{(l)}(x) - \sigma^{(l)}(x)$ is indispensable for our approach and for the Remez algorithm~\cite{remez1934} to fit high-quality polynomial approximations.
%.
%Similarly, computing the extrema of $\pi(x) - \sigma(x)$ is required to fit high-quality polynomial approximations using the Remez algorithm~\cite{remez1934}.
%
Prior work on verified approximation of $\relu(x) = \max(0, x)$ uses its piecewise linearity~\cite{Kern2025fhe}.
This allows bounding the approximation error $\pi(x) - \relu(x)$ by finding the extrema of the polynomials $\pi(x) - 0$ for $x \leq 0$ and $\pi(x) - x$ for $x \geq 0$.
Finding the extrema of a polynomial $\pi(x)$ in Chebyshev basis reduces to finding the eigenvalues\footnotemark of the Colleague matrix associated with its derivative $\frac{d}{dx}\pi(x)$~\cite{trefethen2012atap}.
\footnotetext{This works assumes the correctness of numerically computed eigenvalues.}
%\todo{place assumption on correct eigenvalues here!}.
%
\newcommand{\ourTargetFunctions}{LAE}%
\newcommand{\ourTargetFunctionsLong}{linear almost everywhere}%
\newcommand{\OURTARGETFUNCTIONS}{Linear Almost Everywhere}%
\newcommand{\ourMonotoneTargetFunctions}{MLAE}%
We extend this approach to activation functions that are \emph{monotonically \ourTargetFunctionsLong{}}:
\begin{definition}[Monotonically \OURTARGETFUNCTIONS{}]
\label{def:mlae_funcitons}
A function $g:\mathbb{R}\to\mathbb{R}$ is \emph{monotonically \ourTargetFunctionsLong{}} iff it has bounded total variation, there are $a_1,a_2,b_1,b_2\in\mathbb{R}$ with
$\lim_{x \to \infty} \left(g\left(x\right) - \left(a_1 x + b_1\right)\right) = \lim_{x \to -\infty} \left(g\left(x\right) - \left(a_2 x + b_2\right)\right) = 0$ and there are $c_1,c_2\in\mathbb{R}$ such that
$g\left(x\right)-\left(a_1 x + b_1\right)$ (resp. $g\left(x\right)-\left(a_2 x + b_2\right)$)
is monotonic on $\left(-\infty,c_1\right]$ (resp. $\left[c_2,\infty\right)$).
\end{definition}

\begin{textAtEnd}
%
%Our approach can be applied to networks with any activation function that is \emph{\ourTargetFunctionsLong}:\todo{merge LAE and MLAE and move to 4.2}
If a function satisfies all requirements of Definition~\ref{def:mlae_funcitons} except for the monotonicity requirement, we call it linear almost everywhere (LAE).
We can then derive the following result which enables us to approximate its outermost regions with the linear functions representing the asymptotes:
% In principle, our approach is applicable to functions which are \ourTargetFunctionsLong{}
% \begin{definition}[\OURTARGETFUNCTIONS]
% A function $g:\mathbb{R}\to\mathbb{R}$ is \emph{\ourTargetFunctionsLong} (\ourTargetFunctions{}) if it has bounded total variation (see Section \ref{ssec:total-variation}) and there exist $a_1,a_2,b_1,b_2\in\mathbb{R}$ such that
% $\lim_{x \to \infty} \left(g\left(x\right) - \left(a_1 x + b_1\right)\right) = \lim_{x \to -\infty} \left(g\left(x\right) - \left(a_2 x + b_2\right)\right) = 0$.
% % \begin{align*}
% %     \lim_{x \to \infty} g\left(x\right) - \left(a_1 x + b_1\right) &= 0 &
% %     \lim_{x \to -\infty} g\left(x\right) - \left(a_2 x + b_2\right) &= 0
% % \end{align*}
% \end{definition}
% If a function is \ourTargetFunctions{}, this enables us to approximate its outermost regions with the linear functions representing the asymptotes (see Proposition \ref{prop:method:arbit_lin_approx}).
% In this work, we focus on \emph{monotonically \ourTargetFunctions{} functions} which enables us to easily identify the region where $g$ is not adequately represented by its asymptotes.
% %Approximating \ourTargetFunctions{} functions in practice requires us to identify a region in which $g\left(x\right)$ is not adequately represented by its asymptotes.
% %To achieve this, our work focuses on \emph{monotonically \ourTargetFunctions{} functions}:
% %
% Functions which are \ourTargetFunctions{} admit linear approximation with arbitrary precision $\epsilon^*$ for the intervals approaching infinity:
\begin{proposition}[Arbitrarily Precise Linear Approximation]
\label{prop:method:arbit_lin_approx}
If a function is \ourTargetFunctions{}, then for any $\epsilon^*$ there exists a $\delta < 0$ (resp. $\delta > 0$) such that $\left|g\left(x\right)-\left(a_i x + b_i\right)\right| \leq \epsilon^*$ for all $x \in \left(-\infty,\delta\right]$ for $i=1$ (resp. for all $x \in \left[\delta,\infty\right)$ for $i=2$).
\end{proposition}
\begin{proof}
This result follows directly from the definition of function limits which states $g\left(x\right)-a_i x + b_i$ has limit $0$ iff the following formula is valid:
\[
\forall \epsilon > 0\;\exists c > 0\;\forall x\enspace
x > c \rightarrow \left|\left(g\left(x\right)-a_i x + b_i\right) - 0 \right| < \epsilon.
\]
For any given $\epsilon^*$ we thus choose $\delta$ as the $c$ from the formula above.
\end{proof}

% \begin{table}[]
%     \centering
%     \begin{tabular}{l|l l|l l|l}
%          Activation Function & $a_1$ & $b_1$ & $a_2$ & $b_2$ & \ourMonotoneTargetFunctions{}\\\midrule
%          Sigmoid~\cite[3.2]{kunc2024decades} &
%          $0$ & $0$ &
%          $0$ & $1$ &
%          yes\\
%          Tanh~\cite[3.2]{kunc2024decades} &
%          $0$ & $-1$ &
%          $0$ & $1$ &
%          yes\\\midrule
%          ReLU~\cite[3.6]{kunc2024decades} &
%          $0$ & $0$ &
%          $1$ & $0$ &
%          yes\\
%          Leaky ReLU~\cite[3.6.2]{kunc2024decades} &
%          $1/\alpha$ & $0$ &
%          $1$ & $0$ &
%          yes\\\midrule
%          GeLU~\cite[3.3.1]{kunc2024decades} & 
%          $1$ & $0$ &
%          $0$ & $0$ &
%          yes \\
%          SiLU~\cite[3.3]{kunc2024decades} & 
%          $1$ & $0$ &
%          $0$ & $0$ &
%          yes \\
%          ELU~\cite[3.6.48]{kunc2024decades} & 
%          $1$ & $0$ &
%          $0$ & $-1$ &
%          yes \\
%     \end{tabular}
%     \caption{Overview on \ourTargetFunctionsLong{}}
%     \label{tab:MLAE-activation functions}
% \end{table}

\begin{lemma}[GeLU is MLAE]\label{lemma:gelu-mlae}
    The function $\gelu\left(x\right)=x\cdot\mathbf{\Phi}\left(x\right)$ (for $\mathbf{\Phi}\left(x\right)$ being the gaussian error function) is monotonically linear almost everywhere.
\end{lemma}
\begin{proof}
    We know that $0 \leq \mathbf{\Phi}\left(x\right) \leq 1$ and therefore $0 \leq 1-\mathbf{\Phi}\left(x\right) \leq 1$.

    Additionally, it is easily derivable that:
    \begin{align*}
    1-\mathbf{\Phi}\left(x\right) &= \int_x^\infty \phi\left(u\right) du = \int_x^\infty \frac{u}{u}\phi\left(u\right) du\\
    &\leq \int_x^\infty \frac{u}{x}\phi\left(u\right) du =
    \frac{1}{x}\int_x^\infty u\phi\left(u\right) du
    \end{align*}
    Since $\phi'(x)=-x\cdot\phi(x)$, the last term can be simplified to $\frac{1}{x}\left[-\phi(u)\right]_x^\infty$ which yields $\frac{\phi\left(x\right)}{x}$ as $\phi$'s limit is $0$ for $x\to\pm\infty$.

    We then know that for $x<0$ it holds that $\left|\gelu\left(x\right)\right|=\left|x\cdot\mathbf{\Phi}\left(x\right)\right|=\left(-x\right)\cdot(1-\mathbf{\Phi}\left(-x\right))\leq \left(-x\right)\frac{\phi\left(-x\right)}{-x}=\phi\left(-x\right) \to 0$ (for $x\to -\infty$).
    Further, since $\gelu\left(x\right)<0$ for $x<0$, this yields the limit for $-\infty$ via squeezing.

    We can also show that $\gelu\left(x\right)-x = x\left(\mathbf{\Phi}\left(x\right)-1\right)=-x\left(1-\mathbf{\Phi}\left(x\right)\right)
    \geq (-x)\frac{\phi\left(x\right)}{x} = -\phi\left(x\right)
    $.
    Since $-x(1-\mathbf{\Phi}\left(x\right)) \leq 0$ for $x \geq 0$ and $-\phi\left(x\right) \to 0$ (for $x \to \infty$) this yields that $\lim_{x\to\infty} \gelu\left(x\right)-x=0$ via squeezing.

Also, let $0 > x^* = \arg\min_{x \in \R} \gelu(x)$, then for $x \in (-\infty, x^*]$, $\gelu(x)$ is monotonically decreasing.

Additionally, we have
\begin{align*}
    \gelu(-x) &= (-x) \cdot \Phi(-x) \\
    &= (-x) \cdot (1 - \Phi(x)) = \gelu(x) - x \;.
\end{align*}
Therefore, $\gelu(x) \leq x$ for $x \in [0, \infty)$ and $\gelu(x) - x$ is monotonically decreasing for $x \in [-x^*, \infty)$.
\end{proof}

%\todo{Conjecture: Any S-shaped activation function?}
\end{textAtEnd}%

\noindent
Many activation functions (e.g. GELU, ELU or Sigmoid) satisfy this property (for $\gelu$, we give a proof in Lemma~\ref{lemma:gelu-mlae} in the Appendix).
We approximate functions of this type via piecewise polynomial approximations of the form
\begin{align*}
q(x)=&\mathds{1}_{x \leq d_0}(x)\cdot(a_1 x + b_1) + \mathds{1}_{x > d_n}(x)\cdot (a_2 x + b_2) +\\
&\sum_{i=1}^{n} \mathds{1}_{x \in \left(d_{i-1}, d_i\right]}(x) \cdot q_i\left(x\right)
\end{align*}
%
%In order to extend this approach to other activation functions, we propose to compute piecewise polynomial approximations 
% \begin{align}
%     q(x) = \begin{cases}
%         a_1 x + b_1 &, x \in (-\infty, d_0] \\
%         q_i(x) &, x \in (d_{i-1}, d_i], ~1 \leq i \leq n \\
%         a_2 x + b_2 &, x \in [d_n, \infty)
%     \end{cases}
% \end{align}
with known approximation error $\epsilon_q \geq | q(x) - \sigma(x) |$ for all $x \in \R$ (with %, where $a_j x + b_j$ are the asymptotes and
$d_0 \leq c_1, d_n \geq c_2$, and the indicator function $\mathds{1}_{p(x)}(x)\in\{0,1\}$ being $1$ for $x$ iff $p(x)$).
%relate to the points after which $\sigma$ monotonically converges to these asymptotes. 
%
We can then use $q(x)$ as a surrogate in the Remez algorithm to fit a polynomial $\pi^{(l)}(x) \approx q(x) \approx \sigma^{(l)}(x)$ and bound its approximation error
\begin{align}
    | \pi(x) - \sigma(x) | \leq | \pi(x) - q(x) | + | q(x) - \sigma(x) | \leq \epsilon_\pi + \epsilon_q \;,
\end{align}
where finding $\epsilon_\pi$ reduces to $n+1$ computations of polynomial extrema and $\epsilon_q$ is a known constant.

We can always construct such a function $q(x)$ given a target tolerance $\epsilon_q$.
To this end, we first find points $d_0 \leq c_1, d_n\geq c_2$ with $| a_i x + b_i - \sigma(x) | \leq \epsilon_q$ via binary search.
%
%We use monotonic convergence towards the asymptotes to find points $d_0, d_n$ with $| a_i x + b_i - \sigma(x) | \leq \epsilon_q$ via binary search.
Then, we use Chebyshev approximation to find polynomials $q_i(x) \approx \sigma(x)$ over partitions of $[d_0, d_n]$.
Given concrete $q_i(x)$, we use interval arithmetic~\cite{Moore1966Interval} to check if $| q_i(x) - \sigma(x) | \leq \epsilon_q$ over the interval of interest.
If the target tolerance is violated, we increase the number of sub-intervals.
While this last step is expensive, it only has to be computed \emph{once} for each activation function.

\subsection{Differential Bounds via Differential Verification}
\label{ssec:diff-verification}

%\todo{Text: find ways to reduce it a bit}
Construction of $f_\pi$ according to Section \ref{ssec:range-analysis} guarantees robustness, but we also desire bounds on the deviation $\| f_\pi(\vec{x}) - f(\vec{x}) \|_\infty$ between the polynomial and the original NN. %
Prior work shows that standard verifiers struggle to tightly bound the difference between two NNs~\cite{Paulsen2020Reludiff,Paulsen2020Neurodiff,Teuber2025verydiff}.
%\todo{What are $\net_1$ and $\net_2$?} -- Samuel: Removed
However, when NNs have the same architecture and similar weights, approaches for differential verification~\cite{Paulsen2020Reludiff,Paulsen2020Neurodiff,Banerjee2024Raven,Teuber2025verydiff} produce tighter bounds by
propagating local bounds on the difference through the NN (see Section \ref{sssec:diff-verification}).
%maintaining information about intermediate neuron-wise differences between NNs.

Differential verification %approaches 
leverages linear over-approximations for activation functions.
%are also based on linear overapproximation of the networks' activation functions \todo{should we introduce that earlier somewhere?}.
While existing work provides relaxations for the bivariate non-linearity $\sigma_1^{(l)}(x) - \sigma_2^{(l)}(x - \Delta)$ for $\sigma_1^{(l)} = \sigma_2^{(l)}$ with bounded derivative~\cite{Banerjee2024Raven}, it cannot handle polynomial $\sigma_1^{(l)}$ for $\sigma_2^{(l)}\neq\relu$.
Based on the relaxation by Banerjee \textit{et al.}~\cite{Banerjee2024Raven}, we extend the differential verification tool \textsc{VeryDiff}~\cite{Teuber2025verydiff} with parallel linear relaxations 
%for this case (proof see page \pageref{proof:p-sigma-relaxation}):
for the case where one NN contains an activation function with bounded derivative and the other contains a polynomial activation (proof see page \pageref{proof:p-sigma-relaxation}):

\begin{theoremE}[$p$-$\sigma$ Relaxation][end, text link={}]\label{thm:p-act-relax}
    Given bounds $x \in [l_x, u_x], y \in [l_y, u_y], \Delta \in [l_\Delta, u_\Delta]$, activation function $\sigma(x)$ with derivative bounds $\partial_l \leq \frac{d}{dx} \sigma(x) \leq \partial_u$ for all $x \in [\min(l_x, l_y), \max(u_x, u_y)]$ and polynomial $p(x)$ with $\epsilon = \max_{x \in [l_x, u_x]} | p(x) - \sigma(x)|$, the activation difference is bounded by
    \begin{align*}
        \alpha \cdot \Delta + \beta_l - \epsilon \leq p(x) - \sigma(x - \Delta) \leq \alpha \cdot \Delta + \beta_u + \epsilon \;,
    \end{align*}
    where $\alpha = (\alpha_l + \alpha_u)/2$ and
    \begin{align*}
        %\alpha &= \frac{1}{2} (\alpha_l + \alpha_u), &
        \beta_l &= 
        \begin{cases}
            \lambda_l l_\Delta + \hat{\beta}_l &\!\!\!\!, \lambda_l \geq 0\\
            \lambda_l u_\Delta + \hat{\beta}_l &\!\!\!\!, \lambda_l < 0
        \end{cases}, &\!\!\!\!
        \beta_u &= \begin{cases}
            \lambda_u l_\Delta + \hat{\beta}_u &\!\!\!\!,\lambda_u \leq 0\\
            \lambda_u u_\Delta + \hat{\beta}_u &\!\!\!\!,\lambda_u > 0 \;,
        \end{cases}
    \end{align*}
    The $\alpha_l, \alpha_u, \hat{\beta}_l, \hat{\beta}_u$ are the parameters of the non-parallel linear relaxations given by \textsc{RaVeN}~\cite{Banerjee2024Raven} as
    \begin{align*}
        (\alpha_l, \alpha_u) &= \begin{cases}
            (\partial_l, \partial_u) &, l_\Delta \geq 0\\
            (\partial_u, \partial_l) &, u_\Delta \leq 0\\
            \left(\frac{\partial_l u_\Delta - \partial_u l_\Delta}{u_\Delta - l_\Delta}, 
                  \frac{\partial_u u_\Delta - \partial_l l_\Delta}{u_\Delta - l_\Delta}\right), & \text{otw.}
        \end{cases},\\
        % \alpha_l &= \begin{cases}
        %     \partial_l &, l_\Delta \geq 0\\
        %     \partial_u &, u_\Delta \leq 0\\
        %     \frac{\partial_l u_\Delta - \partial_u l_\Delta}{u_\Delta - l_\Delta}, & \text{otw.}
        % \end{cases}, &
        % \alpha_u &= \begin{cases}
        %     \partial_u, & l_\Delta \geq 0\\
        %     \partial_l, & u_\Delta \leq 0\\
        %     \frac{\partial_u u_\Delta - \partial_l l_\Delta}{u_\Delta - l_\Delta}, & \text{otw.}
        % \end{cases},\\
        \lambda_l &= (\alpha-\alpha_l)\;,
        %&
        \quad
        \lambda_u = (\alpha-\alpha_u)\;,\\
        \hat{\beta}_l &= -\hat{\beta}_u = \frac{(\partial_u - \partial_l) l_\Delta u_\Delta}{u_\Delta - l_\Delta}\;.
        % \begin{array}{ll}
        %      \\
        %      \lambda_l &= (\alpha-\alpha_l),\\
        %      \lambda_u &= (\alpha-\alpha_u)\;.
        % \end{array}
    \end{align*}
\end{theoremE}
%\todo{LR: Is $\lambda_u = (\alpha_u - \alpha_u)$ a typo?} -- Samuel: Yes, fixed now (alpha-alpha_u)
\begin{textAtEnd}
    Soundness of our linear relaxation depends on the relaxation for activation functions with bounded derivative given by \textsc{RaVeN}~\cite[Appendix G.2]{Banerjee2024Raven}.
    % We restate their soundness result here:
    %
    % \begin{lemma}[Correctness of \textsc{RaVeN} Relaxation~\cite{Banerjee2024Raven}]\label{lemma:raven-relax}
    %     Given bounds $x \in [l_x, u_x], y \in [l_y, u_y], \Delta \in [l_\Delta, u_\Delta]$ and $\partial_l \leq \frac{d}{dx} \sigma(x) \leq \partial_u ~~ \forall x \in [\min(l_x, l_y), \max(u_x, u_y)]$, then 
    %     \begin{align}
    %         \alpha_l \cdot \Delta + \beta_l \leq \sigma(x) - \sigma(x - \Delta) \leq \alpha_u \cdot \Delta + \beta_u \;,
    %     \end{align}
    %     where
    %     \begin{align*}
    %         \alpha_l &= \begin{cases}
    %             \partial_l &, l_\Delta \geq 0\\
    %             \partial_u &, u_\Delta \leq 0\\
    %             \frac{\partial_l u_\Delta - \partial_u l_\Delta}{u_\Delta - l_\Delta}, & \text{otw.}
    %         \end{cases}, &
    %         \alpha_u &= \begin{cases}
    %             \partial_u, & l_\Delta \geq 0\\
    %             \partial_l, & u_\Delta \leq 0\\
    %             \frac{\partial_u u_\Delta - \partial_l l_\Delta}{u_\Delta - l_\Delta}, & \text{otw.}
    %         \end{cases}, &
    %         \beta_l &= -\beta_u = \frac{(\partial_u - \partial_l) l_\Delta u_\Delta}{u_\Delta - l_\Delta} \;.
    %     \end{align*}
    % \end{lemma}
    %
    We extend the relaxation introduced by Banerjee \textit{et al.}~\cite{Banerjee2024Raven} to a differential transformer for $p(x) - \sigma(x - \Delta)$ for polynomial $p$ and activation function $\sigma$ using a verified bound on the approximation error $\epsilon \geq \max_{x \in [l, u]} | p(x) - \sigma(x)|$.
    Given this, we can overapproximate the difference $\delta=p(x) - \sigma(x - \Delta)$ by
    \begin{align}
        \sigma(x) - \sigma(x - \Delta) - \epsilon \leq \delta \leq \sigma(x) - \sigma(x - \Delta) + \epsilon 
        \label{eq:poly-act}
    \end{align}
    and if $\sigma$ has bounded derivative, we can therefore apply the corresponding relaxation for $f = g = \sigma$ given in \textsc{RaVeN}~\cite{Banerjee2024Raven}.
    
    However, since \textsc{VeryDiff}~\cite{Teuber2025verydiff} requires parallel linear relaxations, we have to slightly adjust the non-parallel \textsc{RaVeN}-relaxation.
\end{textAtEnd}
\begin{proofE}
    \label{proof:p-sigma-relaxation}
    We construct parallel lower and upper linear relaxations for $\sigma(x) - \sigma(x - \Delta)$ by slightly modifying the non-parallel \textsc{RaVeN} relaxation~\textsc{RaVeN}~\cite[Appendix G.2]{Banerjee2024Raven}:
    \begin{align}
        \alpha_l \cdot \Delta + \hat{\beta}_l \leq \sigma(x) - \sigma(x - \Delta) \leq \alpha_u \cdot \Delta + \hat{\beta}_u \;.
    \end{align}
    We first compute the mean slope
    \begin{align}
        \alpha = \frac{1}{2} (\alpha_l + \alpha_u)
    \end{align}
    of the non-parallel linear lower and upper relaxations.
    Then we compute upward and downward shifts for the linear function $\alpha \cdot \Delta$ s.t. it is larger than \textsc{RaVeN}'s upper relaxation $\alpha_u \cdot \Delta + \beta_u$ and smaller than \textsc{RaVeN}'s lower relaxation $\alpha_l \cdot \Delta + \beta_l$.
    Finally, we apply Equation~\ref{eq:poly-act} to lift the correctness result from $\sigma(x) - \sigma(x - \Delta)$ to a result for $p(x) - \sigma(x - \Delta)$.
    
    For the upward shift, we compute
    \begin{align*}
        \beta_u &= \max_{\Delta \in [l_\Delta, u_\Delta]} \alpha \Delta - (\alpha_u \Delta + \hat{\beta}_u) \\
        &= \max_{\Delta \in [l_\Delta, u_\Delta]} (\alpha - \alpha_u) \Delta + \hat{\beta}_u \\
        &= \begin{cases}
            (\alpha - \alpha_u) l_\Delta + \hat{\beta}_u &,\alpha - \alpha_u \leq 0\\
            (\alpha - \alpha_u) u_\Delta + \hat{\beta}_u &,\alpha - \alpha_u > 0 \;,
        \end{cases}
    \end{align*}
    where the case distinction is due to changing monotonicity of $(\alpha - \alpha_u) \Delta + \hat{\beta}_u$.
    
    Similarly, we compute the shift for the lower bound as
    \begin{align*}
        \beta_l &= \min_{\Delta \in [l_\Delta, u_\Delta]} \alpha \Delta - (\alpha_l \Delta + \hat{\beta}_l)\\&= \min_{\Delta \in [l_\Delta, u_\Delta]} (\alpha - \alpha_l) \Delta + \hat{\beta}_l \\
        &= \begin{cases}
            (\alpha - \alpha_l) l_\Delta + \hat{\beta}_l &, \alpha - \alpha_l \geq 0\\
            (\alpha - \alpha_l) u_\Delta + \hat{\beta}_l &, \alpha - \alpha_l < 0
        \end{cases}
    \end{align*}

    With the above steps, we obtain 
    \begin{align}
        \alpha \Delta + \beta_l \leq \sigma(x) - \sigma(x - \Delta) \leq \alpha \Delta + \beta_u
    \end{align}
    over the domain of interest.
    Applying Equation~\ref{eq:poly-act}, we obtain our final statement
    \begin{align}
        \alpha \Delta + \beta_l - \epsilon \leq p(x) - \sigma(x - \Delta) \leq \alpha \Delta + \beta_u + \epsilon \;.
    \end{align}
\end{proofE}

\section{Polynomial Approximation and Smoothness}
\label{sec:poly-approx-smooth}
The techniques presented in the previous section enable the certified construction of polynomially approximated networks as well as computing verified bounds on the difference to the original network.% for a wide range of activation functions.
These verified bounds can be significantly larger than bounds estimated from a dataset and may thus result in polynomial networks that -- although free from overflows -- poorly approximate the original networks.

We recover from this loss of precision by choosing networks with smooth activation functions like $\gelu$ as the basis for our polynomial networks.
Our evaluation (Section~\ref{sec:experiments}) will show that this allows our \emph{verified} polynomially approximated networks to eventually achieve smaller approximation error, than polynomial approximation of $\relu$ networks based on sampled preactivation bounds.

This is based on the fact that the polynomial approximation error of an activation $\sigma: \R \to \R$ converges faster to zero, the smoother the activation function is.
The following bounds given by Trefethen \textit{et al.}~\cite{trefethen2012atap} for the error of degree $n$ Chebyshev polynomial approximations $\pi_n(x) \approx \sigma(x)$ (w.l.o.g. $x \in [-1, 1]$) formally capture this phenomenon:
%\begin{minipage}{.5\linewidth}
%    \begin{equation}
%        | \sigma(x) - \pi_n(x) | \leq \frac{4V}{\pi \nu (n - \nu)^\nu}
%        \label{eq:diffable}
%    \end{equation}
%\end{minipage}%
%\begin{minipage}{.5\linewidth}
%    \break 
%    \begin{equation}
%        | \sigma(x) - \pi_n(x) | \leq \frac{4 M \rho^{-n}}{\rho - 1}
%        \label{eq:analytic}
%    \end{equation}
%\end{minipage}
\begin{align}
    | \sigma(x) - \pi_n(x) | &\leq \frac{4V}{\pi \nu (n - \nu)^\nu}
    \label{eq:diffable} \\
    | \sigma(x) - \pi_n(x) | &\leq \frac{4 M \rho^{-n}}{\rho - 1}
    \label{eq:analytic}
\end{align}
The first bound, implies that the approximation error decays with rate $O(n^{-\nu})$ and is applicable, if $\sigma$ is $\nu$ times differentiable and $n > \nu$.
Equation~\eqref{eq:analytic} is applicable to activations that are analytic, i.e. their Taylor expansion converges to $\sigma$ for any $x \in [-1, 1]$, and leads to an even faster decay with rate $O(\rho^{-n})$, where a larger $\rho$ for the Bernstein ellipse $E_\rho$ indicates that $\sigma$ can be analytically continued further into the complex plane.
The constants $V$ and $M$ in the above equations are independent of the polynomial degree $n$.
The former is given by the total variation $V = V_{[-1,1]}(\sigma^{\nu})$ of the $\nu$-th derivative of $\sigma$, where 
\begin{align}
    V_{[a,b]}(\sigma) = \int_a^b |\sigma'(x)| \;dx
\end{align}
for continuously differentiable $\sigma$ (see e.g.~\cite{appell2014boundedvariation}).
The latter constant $M = \max_{x \in E_\rho} |\sigma(x)|$ is defined using a parameter $\rho$, s.t. $\sigma$ can be analytically continued to the interior of the Berstein ellipse
\begin{align*}
    E_\rho = \left\{ z \in \mathbb{C} ~\middle|~ z = \frac{\rho e^{i \theta} + \left(\rho e^{i \theta}\right)^{-1}}{2},\; \rho \geq 1,\; 0 \leq \theta < 2\pi \right\}
\end{align*}
(definition of $E_\rho$ taken from Wang \textit{et al.}~\cite{wang2018bernstein}).
Bernstein ellipses for varius choices of $\rho$ are illustrated in Figure~\ref{fig:bernstein-ellipse}.
%\todo[inline]{Add a caption to Fig 4}
If $\sigma$ is \emph{entire} (i.e. it can be analytically continued to all $\mathbb{C}$), then, one can choose any parameter $\rho$ for the Bernstein ellipse that leads to a favourable value of $M = \max_{x \in E_\rho} |\sigma(x)|$ in Equation~\eqref{eq:analytic}.

%\begin{figure}
%    \vspace*{-1.2em}
%    \begin{center}
%    \includegraphics[width=\linewidth]{images/convergence/convergence_approximation_error.pdf}
%    \end{center}
%    \vspace*{-1.3em}
%    \caption{Polynomial approximation error for different activation functions over $x \in [-1, 1]$: Smooth functions converge faster.}
%    \label{subfig:convergence}
%\end{figure}
The effects of these bounds are visualised in Figure~\ref{fig:convergence}:
The approximation error for the once-differentiable $\elu$ activation converges faster than the error for the non-differentiable $\relu$ function. %, which is not differentiable at all.
The errors for the analytic functions $\text{Sigm}(x) = 1/(1 + e^{-x})$ and $\silu(x) = x \cdot \text{Sigm}(x)$ converge much faster.
While $\text{Sigm}(x)$ and $\silu(x)$ have poles at $x = \pm \pi i$, $\gelu(x)$ can be analytically continued to the entire complex plane.
Therefore, we see even faster convergence for this activation function -- at least before convergence of the approximation error levels off due to limits of floating point accuracy.
Consequently, $\silu$ and even more so $\gelu$ makes precise polynomial approximation of neural networks easier.

Figure~\ref{fig:convergence-radius} illustrates how increases in width of the approximation interval affect the polynomial approximation error, for fixed polynomial degrees.
We can see that for smaller radii $r$ of the approximation interval $[-r, r]$, the $\gelu$ approximation has a much smaller error than the $\relu$ approximation, whereas the difference gets smaller as the radius increases.
The Figure further shows the benefits of fitting polynomial approximations using the Remez algorithm~\cite{remez1934} as opposed to standard Chebyshev interpolation.
High-quality polynomial approximations computed by the Remez algorithm are roughly twice as precise as Chebyshev approximations.

\section{Concrete CKKS Procedures}
\label{sec:concrete-ckks}
\begin{figure*}
    \centering
    \includegraphics[width=\linewidth]{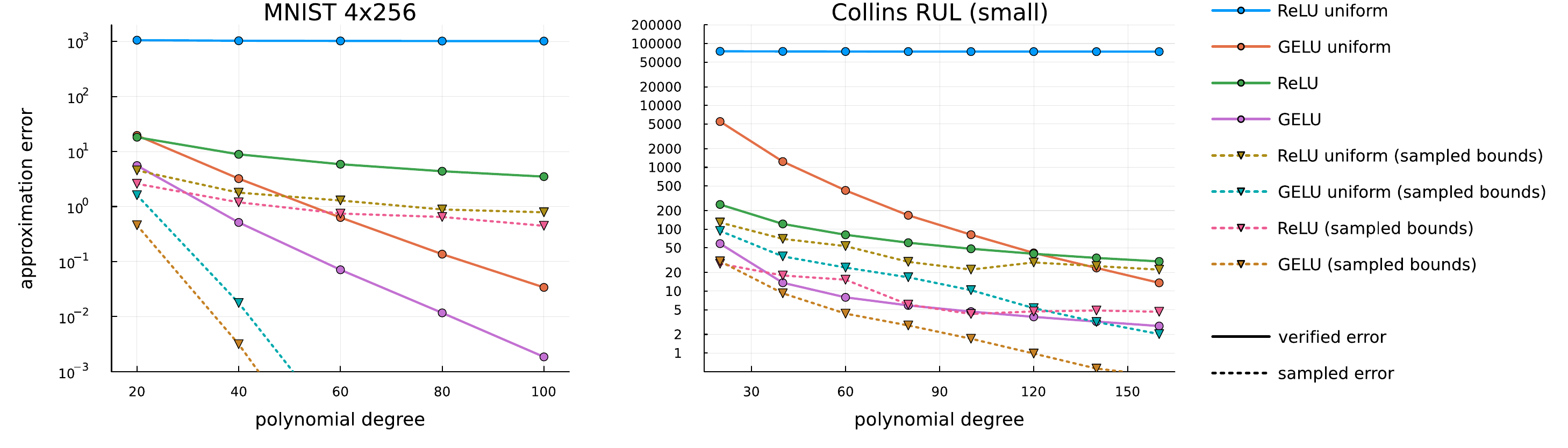}
    \caption{Verified bound on the error $| f_\pi(x) - f(x) |$ for networks with verified pre-activation bounds and error values sampled for networks constructed using sampled bounds. We show results for $\relu$ and $\gelu$ networks with heterogeneous polynomials per layer vs. networks with a single \emph{uniform} polynomial per layer.}
    \label{fig:ablation}
\end{figure*}

%\todo[inline]{Add motivation sentence, like beginning of Section \ref{sec:poly-approx-smooth}.}

%{ssec:pw-overapprox}
%[\mathbf{l}^{(l)},\mathbf{u}^{(l)}]

The construction of polynomial activations in Section \ref{ssec:pw-overapprox} implicitly assumes that we can evaluate a different polynomial for each neuron in the same layer $l$. This goes against standard usages of CKKS-based NNs where typically one evaluates a single polynomial approximation of, e.g.\ ReLU, over all the neurons in the same layer~\cite{Ebel2025orion,Lee2022fhe,LLAMACKKS}. However, the per-neuron ranges $[l_i^{(l)},u_i^{(l)}]$ may differ considerably from each other, and thus much can be gained by approximating them independently, as we demonstrate in Section \ref{ssec:ablations}.
%As we already mentioned, at the core of our method there is the evaluation of different polynomials activations for each neuron in the same layer.
%\todo{Where do we discuss this before? If not, we need to rephrase this sentence (ST)}

Here, we show how to evaluate different polynomials in the different slots of a CKKS ciphertexts by simply generalizing the Paterson-Stockmeyer (PS) algorithm.
In particular, Chen, Chillotti and Song \cite[Theorem 1]{CCS19} first gave a PS version adapted for Chebyshev polynomials. We report, in Algorithm \ref{alg:ps}, a modified version of the algorithm which works over Chebyshev bases and \emph{sets of coefficients}. In particular, the algorithm proposed in \cite{CCS19} is simply an instantiation of ours when all the sets of coefficients are equal.

%\todo[inline]{What are the definitions of $q_i(x), r_i(x), d_i(x)$? LR: they are the quotient and remainder of the poly division}
%\todo[inline]{LR: still some sanity check needed}

We refer to Section \ref{sec:cert-nn-to-ckks} for further discussion about the CKKS procedures.

%\todo[inline]{LR: Describe mods to standard CKKS}

\section{Experiments}
\label{sec:experiments}

%We aim to measure the following quantities:
%\begin{itemize}
    %\item Error bounds (diff verif)
    %\item Error bounds (FHE noise)
    %\item Accuracy comparisons
    %\item Certified accuracy comparisons
    %\item Speed comparisons
%\end{itemize}

%\todo[inline]{
%- MNIST
%- CIFAR10
%- HELOC
%- EC
%- NN4Sys
%- Collins RUL
%}
Our experiments\footnote{Refer to \url{https://github.com/lorenzorovida/encrypted-neural-networks-without-overflows}} aim at
evaluating the impact of each contribution on the performance of our approach (Section~\ref{ssec:ablations}),
%(ii)
%showing the existence of overflow events (Section~\ref{sec:design_atk_compare}), and
%shed light on the effect of each of our contributions on the performance of the polynomial NNs (Section~\ref{ssec:ablations}) and
%(iii) 
and demonstrate that certified design (with all our improvements) yields comparable accuracy to non-certified design for both plaintext and encrypted inference (Section~\ref{ssec:large-scale}). %\todo{when using our improvements GeLU and heterogenuous}. %
%demonstrate that certified design of polynomial NNs leads to comparable accuracy in the plaintext setting as well as under encrypted inference (Section~\ref{ssec:large-scale}) 

We run our evaluation on fully connected (FC) NNs trained on the HELOC (Home Equity Line of Credit)~\cite{FICOChallenge}, MNIST~\cite{Lee2022fhe} and EC (Energy Consumption)~\cite{CER2011electricity} datasets, as well as convolutional (Conv) NNs trained on CIFAR10~\cite{krizhevsky2009cifar}.
We further retrained FC NNs based on the NN4Sys benchmark~\cite{Lin2024NN4Sys} for learned index prediction and Conv NNs for remaining useful life prediction (Collins RUL)~\cite{Kirov2023CollinsRUL} taken from the annual NN verification competition \textsc{VNNComp}~\cite{vnncomp2025}.
Overall, the NNs cover a wide range of sizes containing between $64$ and $10300$ neurons.
Following prior work on certified design of polynomial NNs based on $\relu$ NNs~\cite{Kern2025fhe}, we trained all NNs with $L_1$ regularization to encourage tighter pre-activation bounds.

We use a piecewise-polynomial overapproximation $q(x)$ for the $\gelu$ activation function (see Section~\ref{ssec:pw-overapprox}) with verified error $|q(x) - \gelu(x)| \leq 10^{-10}$.
The overapproximation uses $6$ pieces of degree $15$ in addition to the linear asymptotes.

Our implementation is written in Julia~\cite{bezanson2017julia} and builds on top of the differential verifier \textsc{VeryDiff}~\cite{Teuber2025verydiff}.
Execution of polynomial NNs under CKKS is based on OpenFHE~\cite{OpenFHE}.
To compute verified pre-activation ranges, we leverage the state-of-the-art NN verifier $\alpha$-CROWN~\cite{Xu21aCROWN}.
Experiments for plaintext NNs were run on Ubuntu with a $4$-Core Intel Xeon E5-1630v3 and $128$ GiB of RAM.
CKKS inference was run on a M5 Max CPU with $36$ GiB of RAM.
Additional information and results for comparison to the state-of-the-art certified construction approach for $\relu$ neural networks~\cite{Kern2025fhe} can be found in Appendix~\ref{sec:additional-results}.
%Additional evaluation results can be found in Appendix~\ref{sec:additional-results}.

\subsection{From Certified Networks to CKKS Programs}
\begin{figure*}[ht]
    \centering
    \begin{minipage}{0.48\linewidth}
        \centering
        \includegraphics[width=\linewidth]{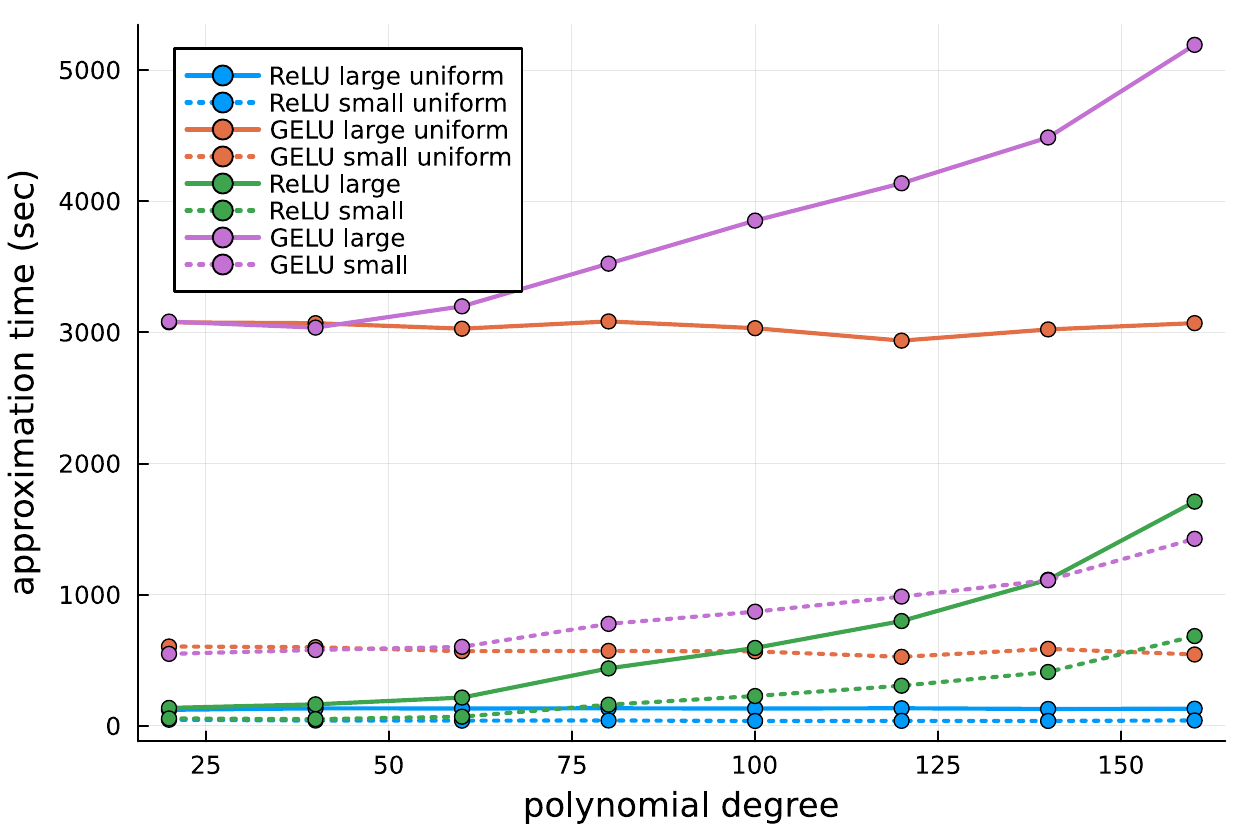}
        \captionof{figure}{Runtime of certified design for Collins RUL.}
        \label{fig:collins-runtime}
    \end{minipage}\hfill
    \begin{minipage}{0.48\linewidth}
        \centering
        \includegraphics[width=\linewidth]{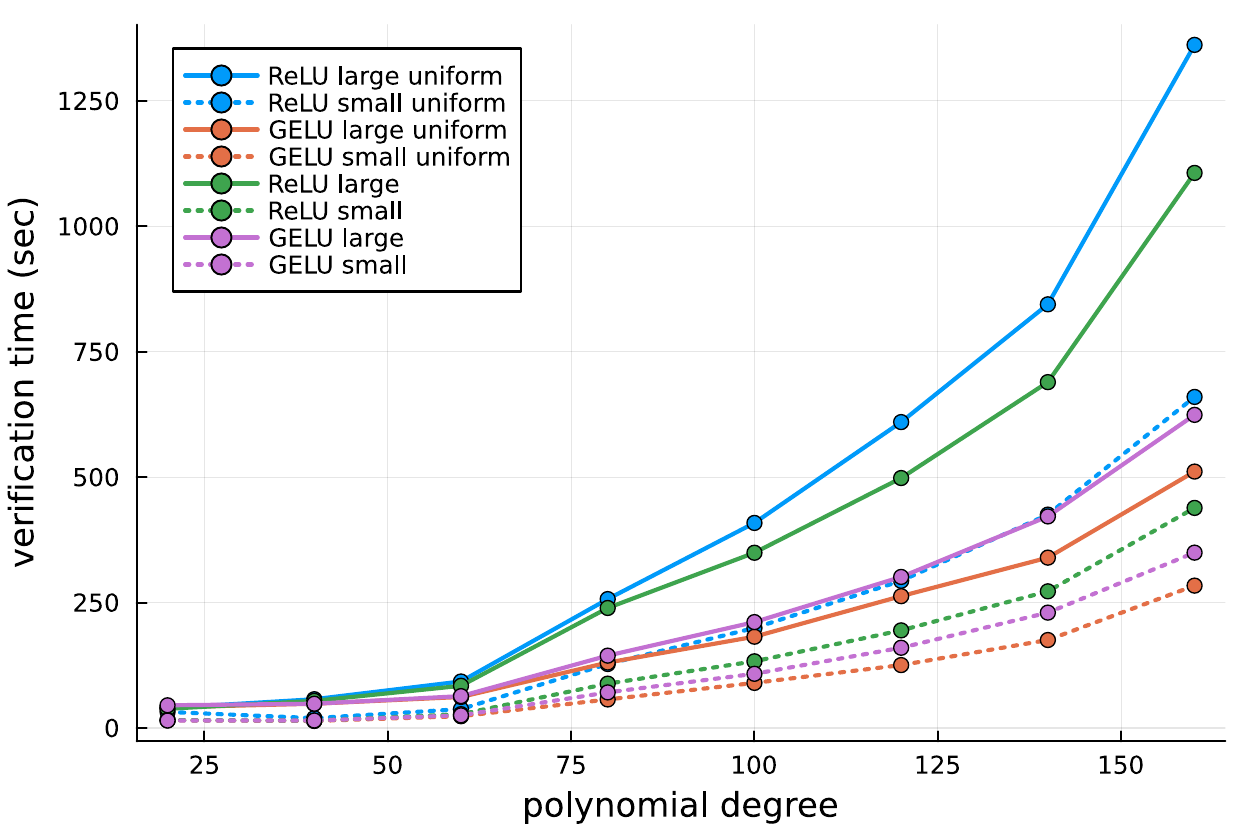}
        \captionof{figure}{Runtime of verified error bound computation for Collins RUL.}
        \label{fig:collins-verification-time}
    \end{minipage}\hfill
    \end{figure*}
%\todo[inline]{LR: Reduce this and comment grammar (put it into the readme GitHub)}
\label{sec:cert-nn-to-ckks}
In order to evaluate our encrypted circuits, we design a compiler that, given a network (defined as a list of linear and Chebyshev layers), generates a JSON file specifying the concrete, certified CKKS circuit.
%that follows the grammar defined in Figure \ref{fig:grammar}.
% Moved grammar here:
%\input{old_draft/grammar}
Our OpenFHE implementation takes the JSON file as input and automatically evaluates the layers in the network.
We give an overview of our methods in Appendix \ref{sec:network-to-json}. We designed our software
to choose the FHE parameters in a way that minimizes the required multiplicative depth
%in such a way that the FHE parameters are automatically chosen to support the minimum multiplicative depth required by the network
and at the same time to support $\lambda > 128$ bits of (classical) security according to \cite{FheGuidelines}. 
\label{sec:ckks_nn_translation}

\subsection{Ablation Study}
\label{ssec:ablations}

% \begin{figure*}[ht]
%     \centering
%     \begin{minipage}{0.32\linewidth}
%         \centering
%         \includegraphics[width=0.95\linewidth]{images/cifar_large_sampled_nan.png}
%         \vspace*{0.25cm}
%         \captionof{figure}{Realistic perturbations of CIFAR images causing overflow events}
%         \label{fig:mnist-perturbations}
%     \end{minipage}
%     \begin{minipage}{0.32\linewidth}
%         \centering
%         \includegraphics[width=\linewidth]{images/experiments/ablation_collins_runtime.pdf}
%         \captionof{figure}{Runtime of certified design for Collins RUL.}
%         \label{fig:collins-runtime}
%     \end{minipage}\hfill
%     \begin{minipage}{0.32\linewidth}
%         \centering
%         \includegraphics[width=\linewidth]{images/experiments/ablation_collins_verification_time.pdf}
%         \captionof{figure}{Runtime of verified error bound computation for Collins RUL.}
%         \label{fig:collins-verification-time}
%     \end{minipage}\hfill
%     %\caption{\textcolor{orange}{Probably ditch one figure or put it somewhere else. Also: How can we get rid of the common caption? ((a),(b) are not too related to (c))}}
% \end{figure*}

%\todo[inline]{Different FHE designs: e.g. multiple polynomial activations per layer, rather than worst-case uniform one.}

%\todo[inline]{Polynomial degree vs error bound size: plots and discussion. Impact on FHE latency. GELU vs ReLU}

%\todo[inline]{Show that FHE noise is ``negligible'': compare errors of normal net vs polynomial net vs FHE net.}

\looseness=-1
We construct polynomial NNs approximating a base NN for multiple values of the polynomial degree $d$ and measure the verifiable bound on the error $\| f_\pi(\vec{x}) - f(\vec{x}) \|_\infty$ as well as the maximum observed error (sampled over the respective dataset).
Comparing the error metrics between different construction settings 
%-- using a base NN with $\relu$ vs. $\gelu$ activation, fitting a single polynomial approximation for all neurons in a layer vs. one polynomial per neuron, using verified pre-activation bounds vs. sampled bounds \todo{is this all we want to show here? In earlier TODOs, we also had: Impact on FHE latency, negligible FHE noise, normal net vs polynomial net vs FHE net} -- 
enables us to evaluate the effect of our individual contributions.
To show the effects on FC and Conv NNs, we report the results for the MNIST $4\times 256$ NN and the small Collins RUL NN ($[1600,2400,800,600,100]$ neurons per layer) in Figure~\ref{fig:ablation}.

%\todo[inline]{PK: The mini-paragraphs below can be tidied up and expanded (we have space now!). Perhaps use bullet points instead?}

\subsubsection*{GELU vs. ReLU Networks}  % \paragraph{} in this template doesn't really create emphasis
Figure~\ref{fig:ablation} clearly shows that the approximation error for $\gelu$ NNs is far smaller than for NNs with $\relu$ activation. 
This effect gets stronger with increasing polynomial degree $d$, reflecting the superior approximation properties of the smooth $\gelu$ activation (see Section~\ref{sec:poly-approx-smooth}).
For large degrees and the Collins benchmark, the error values are at least an order of magnitude smaller for $\gelu$ NNs than for $\relu$ NNs, on the MNIST benchmark the improvement is at least three orders of magnitude for $d = 100$.
This greater advantage in the MNIST case is likely caused by significantly smaller pre-activation bounds for this benchmark.
For larger preactivation bounds, higher polynomial degrees are required to realize the improved convergence of the $\gelu$ polynomial approximation (see also Figure~\ref{fig:convergence-radius}).

\subsubsection*{Heterogeneous vs. Uniform Polynomial Layers}
Error values when using separately fitted polynomial approximations for each neuron are significantly smaller compared to using a single uniform polynomial approximation for all neurons in the same layer.
Exploitation of tighter individual preactivation bounds leads to at least $4.7\times$ smaller approximation error on the Collins benchmark for large degrees.
For verified construction of MNIST networks, we achieve even larger improvements of at least $18\times$ for $d = 100$. 

\subsubsection*{Verified vs. Sampled Bounds}
Using sampled pre-activation bounds leads to substantially smaller sampled error values.
This is expected as sampled values (even if bounds are widened by a constant factor) do not suffer from overapproximation to the same extent as bounds computed using NN verification techniques.
However, these polynomial NNs are prone to overflow events.
Moreover, the above effect is only valid if we fix the activation function:
When we compare the sampled error for $\relu$ NNs with sampling-based bounds to the sampling error for $\gelu$ NNs with verified bounds, the $\gelu$ NNs achieve a better or similar approximation error already for $d \leq 40$ on both MNIST and Collins RUL.
The effect is even more pronounced comparing $\relu$ NNs with a single polynomial per layer (which is the current approach in modern workflows, as in Figure \ref{fig:lit_map}) %\todo{clearly point out that this is the previous state of the art?} 
and sampled bounds to verified $\gelu$ NNs with multiple polynomials per layer.

\subsection{Certified Design at Scale}
\label{ssec:large-scale}

\begin{table*}[th]
\newcommand{\reluSym}{R}
\newcommand{\geluSym}{G}
\renewcommand{\arraystretch}{1.16} % Default 
\setlength{\tabcolsep}{4pt} % Default value: 6pt
    %\caption{Verified error bounds and comparison of performance metrics between base networks and their polynomial approximations using verified pre-activation bounds.
    %The correct number of digits (matching results for $f_\pi$ and $f_{CKKS}^\bot$) are computed as $L_\infty|{\bm e}_\text{\normalfont CKKS}|$, where ${\bm e}_\text{CKKS}$ is final CKKS error.\vspace*{1mm}}
    \caption{Verified error bounds and cryptographic error of certified FHE inference. We report the $L_1$ penalty used for training the original NN, $\mathbf{\sigma}$ is the activation function ($\reluSym \equiv \relu{}, \geluSym\equiv\gelu{}$),
    $N$ is the ring size, $QP$ is the modulus size. $D$ refers to the multiplicative depth of the circuit. The (max) CKKS error is computed as $| f_{\text{CKKS}}(x) - f_\pi(x) |$ obtained over a sample of the respective dataset. Latencies ($t$) are in seconds and have been measured on a M5 Max CPU with 36GB of memory. Moduli with $^*$ are reported although their security is slightly smaller than 128-bits. Refer to \cite[Section 2.3]{KPP22} for the definition of $d_\text{num}$\vspace*{2mm}}
    \label{tab:large-scale-long}
    \label{tab:large-scale}
    \footnotesize
    \centering
    % for wide source-code view this table is nicely formatted
    \begin{tabular}{llllll rrr llllllll}
         & \textbf{Size}        & $L_1$      & $\mathbf{\sigma}$  & \begin{tabular}[c]{@{}c@{}}\textbf{Output range}\\\textbf{(sampled)} \end{tabular}    & \begin{tabular}[c]{@{}c@{}}\textbf{Acc, MSE,}\\\textbf{MAE of $f$} \end{tabular}  & \makecell[l]{\textbf{Poly.}\\\textbf{deg.}}    & \begin{tabular}[c]{@{}c@{}}\textbf{Acc, MSE,}\\\textbf{MAE of $f_{\pi}$} \end{tabular}      & \begin{tabular}[c]{@{}c@{}}\textbf{Verified}\\\textbf{error} \end{tabular}    & \begin{tabular}[c]{@{}c@{}}\textbf{CKKS}\\\textbf{error} \end{tabular}
&
$N$ & $QP$ & $q_i$ & $\Delta$ & $d_\text{num}$ &
$D$ & 
$t$
\\\midrule
    %%% HELOC
          \parbox[t]{2mm}{\multirow{2}{*}{\rotatebox[origin=c]{90}{\textbf{HEL.}}}}          
          & $[64,32]$         & 2e-5 & \reluSym        &
          $[-5.3, 3.6]$
          & $73.31\%$     & $27$      & $73.24\%$           & $0.75$            & $10^{-9}$ &
          $2^{16}$ & $2^{1306}$       &  45 & 60 & 2 & 21 & 1
          \\ 
        &                   &  & \geluSym          &         
        $[-4.5, 2.9]$
        & $72.25\%$     & $27$      & $72.25\%$           & 3.76e-6         & $10^{-9}$ &
        $2^{16}$ & $2^{1306}$       &  45 & 60 & 2 & 21 & 1
\\\midrule
    %%% MNIST
         \parbox[t]{2mm}{\multirow{4}{*}{\rotatebox[origin=c]{90}{\textbf{MNIST}}}}          
         & $4 \times 256$    & 1e-4 & \reluSym          & 
         %$[-32.02, 19.65]$
         $[-32.0, 19.7]$
         & $98.64\%$     & $119$     & $98.64\%$           & $2.93$            & $10^{-7}$ &
         $2^{17}$ & $2^{2746}$        &   45       & 60 & 2 & 49 & 39
         \\
        &                   & & \geluSym          & 
        %$[-18.87, 25.03]$
        $[-18.9, 25.0]$
        & $98.93\%$     & $119$     & $98.93\%$           & 3.08e-4         & $10^{-7}$ &
        $2^{17}$ & $2^{2746}$        &   45       & 60 & 2 & 49 & 39
        \\
        & $6 \times 256$    & 1e-4 & \reluSym          &
        %$[-39.37, 25.17]$
        $[-39.4, 25.2]$
        & $99.12\%$     & $119$     & $99.09\%$           & $9.34$            & $10^{-7}$ &
        $2^{17}$ & $2^{3841*}$        &   45       & 60 & 7 & 72 & 74
        \\
        &                   & & \geluSym          &
        %$[-40.33, 16.74]$
        $[-40.3, 16.7]$
        & $98.94\%$     & $119$     & $98.93\%$           & $0.046$           & $10^{-7}$ &
        $2^{17}$ & $2^{3841*}$        &   45       & 60 & 7 & 72 & 74
\\\midrule
    %%% CIFAR10
         \parbox[t]{2mm}{\multirow{8}{*}{\rotatebox[origin=c]{90}{\textbf{CIFAR10}}}}        
         & \multirow{2}{*}{\makecell[l]{$[900,392,$\\$144,256]$}}        & 1e-4 & \reluSym          &
         %$[-16.22, 14.71]$
         $[-16.2, 14.7]$
         & $67.24\%$     & $119$     & $66.18\%$           & $102.47$          & $10^{-4}$ &
         $2^{17}$ & $2^{3511}$ & 45 & 60 & 5 & 62 & 95
         \\
         &                 &   & \reluSym          &                   &               & $247$     & $67.57\%$           & $49.08$           & $10^{-4}$ &
         $2^{17}$ & $2^{3511}$ & 45 & 60 & 7 & 66 & 119
        \\
        &               &    & \geluSym          &
                        %$[-13.93, 17.21]$
                        $[-13.9, 17.2]$
                        & $69.56\%$     & $119$     & $68.70\%$           & $26.98$           & $10^{-4}$ &
                        $2^{17}$ & $2^{3511}$ & 45 & 60 & 5 & 62 & 95
        \\
        &               &    & \geluSym          &                   &               & $247$     & $69.50\%$           & $0.93$            & $10^{-4}$ &
                        $2^{17}$ & $2^{3511}$ & 45 & 60 & 7 & 66 & 119
        \\
        & \multirow{2}{*}{\makecell[l]{$[1.8\mathrm{k},784,$\\$288,256]$}}         & 1e-4   & \reluSym          &
                        %$[-13.33, 16.36]$
                        $[-13.3, 16.4]$
                        & $76.30\%$     & $119$     & $72.35\%$           & $232.48$          & $10^{-4}$    &
                        $2^{17}$ & $2^{3511}$ & 45 & 60 & 5 & 62 & 167
        \\
        &               &    & \reluSym          &                   &               & $247$     & $76.31\%$           & $111.40$          & $10^{-4}$     &
                        $2^{17}$ & $2^{3511}$ & 45 & 60 & 7 & 66 & 191
        \\
        &               &    & \geluSym          &
                        %$[-14.67, 20.28]$
                        $[-14.7, 20.3]$
                        & $77.28\%$     & $119$     & $65.20\%$           & $85.58$           & $10^{-4}$    &
                        $2^{17}$ & $2^{3511}$ & 45 & 60 & 5 & 62 & 167
        \\
        &               &    & \geluSym          &                   &               & $247$     & $75.13\%$           & $18.64$           & $10^{-4}$    &
                        $2^{17}$ & $2^{3511}$ & 45 & 60 & 7 & 66 & 191
\\\midrule
    %%% EC
         \parbox[t]{2mm}{\multirow{2}{*}{\rotatebox[origin=c]{90}{\textbf{EC}}}}         &  $1 \times 64$  & $-$   & \reluSym          &
         %$[0.18, 1.95]$
         $[0.2, 2.0]$
         & $0.0839$      & $27$      & $0.0839$            & $0.11$            & $10^{-6}$    &
         $2^{15}$ &  $2^{781}$    & 45 & 60    & 2 & 12 & 0.5
         \\
        &            &       & \geluSym          &
            %$[0.21, 1.96]$
            $[0.2, 2.0]$
            & $0.0844$      & $27$      & $0.0844$            & 3.64e-7         & $10^{-6}$    &
            $2^{15}$ &  $2^{781}$    & 45 & 60    & 2 & 12 & 0.5
\\\midrule
    %%% NN4Sys
         \parbox[t]{2mm}{\multirow{4}{*}{\rotatebox[origin=c]{90}{\textbf{NN4Sys}}}}   & $3 \times 128$    & 5e-7     & \reluSym          & 
         %$[0.00, 1.01]$
         $[0.0, 1.0]$
         & 4.56e-4     & $59$      & 4.83e-3           & $0.089$           & $10^{-9}$    &
         $2^{17}$ & $2^{2431}$ & 45 & 60  & 2 & 34 &  6
         \\
        &            &       & \geluSym          & 
         %$[0.00, 1.01]$
         $[0.0, 1.0]$
             & 3.57e-4     & $59$      & $3.57e$-$4$         & 1.75e-6         & $10^{-9}$   &
             $2^{17}$ & $2^{2431}$ & 45 & 60  & 2 & 34 &  6
        \\
        & $5 \times 128$  & 5e-7   & \reluSym          & 
         %$[0.00, 1.01]$
         $[0.0, 1.0]$
             & 5.75e-4     & $59$      & 3.10e-3           & $0.11$            & $10^{-9}$    &
             $2^{17}$ & $2^{3136}$ & 45 & 60 & 2 & 55  & 22
        \\
        &           &        & \geluSym          & 
         %$[0.00, 1.01]$
         $[0.0, 1.0]$
             & 2.93e-4     & $59$      & 2.93e-4           & 1.15e-5         & $10^{-9}$    &
             $2^{17}$ & $2^{3136}$ & 45 & 60 & 2 & 55 &  22
\\\midrule
    %%% Collins
         \parbox[t]{2mm}{\multirow{4}{*}{\rotatebox[origin=c]{90}{\textbf{Collins}}}}    
         & \multirow{2}{*}{\makecell[l]{$[1.6\mathrm{k},2.4\mathrm{k},$\\$800,600,100]$}}      & 1e-3       & \reluSym          & 
         %$[0.42, 648.10]$
         $[0.4, 648.1]$ & $457.00$      & $119$     & $456.91$            & $40.43$           & $10^{-1}$   &
         $2^{17}$ & $2^{3511}$ & 45 & 60 & 5 & 61 & 91 
         \\
        &            &       & \geluSym          & 
        %$[0.96, 634.80]$
        $[1.0, 634.8]$
        & $527.89$      & $119$     & $527.93$            & $3.87$            & $10^{-1}$   &
        $2^{17}$ & $2^{3511}$ & 45 & 60 & 5 & 61 & 91
        \\
        & \multirow{2}{*}{\makecell[l]{$[3.2\mathrm{k},4.8\mathrm{k},$\\$1.6\mathrm{k},600,100]$}}   & 1e-3 & \reluSym          & 
    %$[-8.58, 644.74]$
    $[-8.6, 644.7]$
    & $512.60$      & $119$     & $518.46$            & $68.99$           & $10^{-4}$    &
    $2^{17}$ & $2^{3511}$ & 45 & 60 & 5 & 61 & 158
    \\
    &              &     & \geluSym          &
    %$[0.38, 615.81]$
    $[0.4, 615.8]$
    & $699.21$      & $119$     & $686.50$            & $9.28$            & $10^{-4}$    &
    $2^{17}$ & $2^{3511}$ & 45 & 60 & 5 & 61 & 158\\
    \end{tabular}
\end{table*}

In this section, we demonstrate that our certified design approach can be applied to networks of varying sizes from different domains.
As per Section \ref{ssec:ablations}, we compute verified pre-activation bounds and fit  multiple polynomials per layer.
We compare approximations of $\relu$ and $\gelu$ NNs and discuss the quality of the approximated networks as well as the runtime of our approach.

\subsubsection*{Approximation Quality}
For each polynomial network, we measure the verified bound on the approximation error $\| f(\vec{x}) - f_\pi(\vec{x})\|_\infty$ and cryptographic error over a sample of the respective dataset.
We report these values in Table~\ref{tab:large-scale} along with the output range of the NNs sampled over the corresponding dataset to put the approximation error into perspective.
The table also contains traditional metrics of the NNs: 
Classification accuracy (HELOC, MNIST, CIFAR), mean squared error (EC, Collins RUL) and mean absolute error (NN4Sys).

Regarding the traditional metrics (classification accuracy, mean squared error, mean absolute error), the results show only very minor drops in performance between the original NN $f$ and its polynomial approximation $f_\pi$.
The large CIFAR10 NNs with degree $119$ being the lone exception.

Although the verified bound on the approximation error $\| f_\pi(x) - f(x) \|_\infty$ increases with network size, it is still fairly tight for most $\gelu$ NNs -- especially when compared to their typical output ranges.
The same is true for the maximum CKKS error $| f_{\text{CKKS}}(x) - f_\pi(x) |$ of the $\gelu$ networks.
However, polynomial approximations for NNs with $\relu$ activations, always lead to a much worse verified error bound.

\subsubsection*{Runtime}
We analyze the runtime of our approach for both certified construction of a polynomial NN (see Figure~\ref{fig:collins-runtime}), and computation of a verified bound on the approximation error (Figure~\ref{fig:collins-verification-time}) over a range of polynomial degrees.
This experiment was carried out on the Collins RUL benchmark, which contains the largest networks in our benchmark set.

As shown in Figure~\ref{fig:collins-runtime}, the runtime of our certified design approach increases for larger polynomial degrees as computation of polynomial roots is more expensive.
Synthesis of networks with a single uniform polynomial per layer only requires to run the Remez algorithm once per layer, instead of doing so for every neuron, explaining the large discrepancy in runtime.

Independent of the polynomial degree, there is a large cost related to the runtime of the underlying NN verifier, which scales with the size of the network, but also depends on the computational efficiency of computing relaxations for different activation functions.
For \textsc{AutoLiRPA}~\cite{Xu2020autolirpa}, the implementation of the $\alpha$-CROWN NN verification algorithm~\cite{Xu21aCROWN} we use, the relaxations for the $\gelu$ activation function are significantly more expensive to compute than the $\relu$ relaxations.

For computation of the verified bound on the approximation error, we instead use the differential NN verifier \textsc{VeryDiff}~\cite{Teuber2025verydiff} as a backend.
In this setting, we observe faster runtimes for the $\gelu$ networks, since the relaxation described in Theorem~\ref{thm:p-act-relax} for the difference between $\gelu$ and its polynomial approximation does not require any computation of polynomial roots -- in contrast to the utilized $\relu$ difference relaxation~\cite{Kern2025fhe}.
Runtime still increases with the polynomial degree, since overapproximation of the individual polynomial activations requires executing the Remez algorithm. More details about how the CKKS experiments have been performed are given in Section \ref{sec:cert-nn-to-ckks}. 

\section{Conclusions and Directions for Future Work}
\label{sec:conclusions}

\looseness=-1
In this paper, we demonstrate for the first time the vulnerability of current CKKS-based neural networks to overflow events. To solve this, we provide a framework to design certified neural networks. By construction, our approach yields networks that are provably overflow robust on the full input space, without sacrificing predictive performance.
Furthermore, our approach can bound the polynomial approximation error of the plaintext via a differential verification procedure. These bounds complement existing estimates on the magnitude of the cryptographic noise ${\bm e}_\text{CKKS}$, thus providing a complete toolkit to assess the maximum output error introduced during CKKS private inference.
%In particular, going back to Eq.\ref{eq:error_sources}, we can now bound the $f_\pi(\mathbf{x})$ term, and combined with a bound on ${\bm e}_\text{CKKS}$ due to, e.g., $\mathsf{Estimate}$ (Def. \ref{def:correctness}), we have strong tools to bound values in CKKS-based neural networks.

%We believe this can serve as a starting point for new more secure FHE pipelines. 
%An exciting direction for future work is represented by the application of our techniques on CKKS networks that require bootstrapping, thus enabling the evaluation of arbitrary depth networks.

\looseness=-1
\subsubsection*{Limitations and future work} We see this work as a starting point for more robust FHE inference.
Precise modelling of the CKKS noise like in~\cite{CCHMOP23,JCMNT25} could strengthen our bounds to natively include ${\bm e}_\text{CKKS}$ as an additional probabilistic term.
Extending our certified design to arbitrary depth networks is possible in theory but requires the inclusion of bootstrapping and co-design of convolution layers~\cite{10.1145/3658644.3690375} to reduce the CKKS noise.
Without that, we expect our method to produce slack bounds on deeper architectures, such as those commonly used in state-of-the-art image classifiers, thus reducing its accuracy. %\todo{Mention that our accuracies are lower than the sota, and a small discussion about scalability}
Finally, our method may occasionally generate polynomials with ill-scaled coefficients, which cause a reduction on the CKKS precision.

\subsubsection*{Acknowledgements}
%Generative AI was used for editorial purposes in this manuscript, and all outputs were inspected by the authors to ensure accuracy and originality. Further, we used LLMs to facilitate running experiments and implementing standard methods.

The Gemini, Claude, Copilot and Qwen AI models were used for minor editing purposes in this manuscript, such as spellchecking and shortening a few lengthy descriptions. Furthermore, we used them as coding assistants to help with plotting code in Julia and as a starting point in the implementation of the update step (finding points with alternating error) in the Remez algorithm. 
All outputs were inspected by the authors to ensure accuracy of the final output.

%\todo[inline]{Add, as a limitation, the fact that the obtained poly degrees are pretty large and forces us to scale them down by say 0.01. this reduces the final accuracy when the result is scaled back by 100}

%\clearpage
%\todo[inline]{Clean the bibliography.}

\bibliographystyle{IEEEtran}
\bibliography{ref_ieee}

\appendices

\section{Improved GELU Relaxation}

\begin{figure}[ht]
\resizebox{\linewidth}{!}{
\begin{forest}
    for tree={
        nodes={draw, rectangle, inner sep=2pt, align=center, text width=2.5cm},
        grow=south,
        l sep = 1.5cm,  % level separation
        s sep = 0.8cm,  % sibling separation
        yes/.style={edge label={node[midway, left]{yes}}},
        no/.style={edge label={node[midway, right]{no}}},
        if n children=0{
            % leaf nodes with light green background
            fill=green!10,
        }{
            % decision nodes with light blue background
            fill=blue!5
        }
    }
    [$u \leq -\sqrt{2}$
        [concave left, yes]
        [$u \leq \sqrt{2}$, no
            [$l \geq -\sqrt{2}$, yes
                [convex middle, yes]
                [$\gelu'(l) \leq \frac{\gelu(u) - \gelu(l)}{u - l}$, no
                    [non-convex middle increasing, yes]
                    [$\gelu'(u) \leq \frac{\gelu(u) - \gelu(l)}{u - l}$, no
                        [non-convex middle decreasing II, yes]
                        [non-convex middle decreasing I, no]
                    ]
                ]
            ]
            [$l \geq \sqrt{2}$, no
                [concave right, yes]
                [$\gelu'(u) \geq \frac{\gelu(u) - \gelu(l)}{u - l}$, no
                    [outer, yes]
                    [$\gelu'(l) \geq \frac{\gelu(u) - \gelu(l)}{u - l}$, no
                        [non-convex positive II, yes]
                        [non-convex positive I, no]
                    ]
                ]
            ]
        ]
    ]
\end{forest}}
\caption{Case distinction for improved $\gelu$ relaxation.}
\label{fig:improved-gelu-cases}
\end{figure}

In an attempt to compute tighter bounds for Equation~\eqref{eq:nn_optimization} for $\gelu$ activations, we tried to improve upon the parameterized $\gelu$ relaxation given by Shi \textit{et al.}~\cite{Shi2025general} used in $\alpha$-CROWN~\cite{Xu21aCROWN}.
While our relaxation achieves bounds that are significantly tighter at parameter initialization, the difference to the relaxation by Shi \textit{et al.} after parameter optimization is minimal.
Since we used our relaxation in our code, we nevertheless give a brief overview over its construction.

At a high level, we provide linear over- and underapproximations
\begin{align}
    \lo{a}&(\alpha_1, l, u) x + \lo{b}(\alpha_1, l, u) \leq \gelu(x) \\
    \leq&~ \up{a}(\alpha_2, \alpha_3, l, u) + \up{b}(\alpha_2, \alpha_3, l, u) \quad \forall x \in [l, u]\;,
\end{align}
where our relaxation is parameterized by a lower tangent point $\alpha_1$, an upper left tangent point $\alpha_2$ and an upper right tangent point $\alpha_3$.
We then use the case distinction shown in Figure~\ref{fig:improved-gelu-cases} to select if one of the tangents or a secant line is a valid lower or upper bound.

The $\gelu$ function is concave for $x \in (-\infty, -\sqrt{2}]$ and $x \in [\sqrt{2}, \infty)$. 
This relates to the cases \emph{concave left} and \emph{concave right} in Figure~\ref{fig:improved-gelu-cases}.
If $[l, u]$ is completely within these regions, the secant is always the best linear lower bound and the tangent at any point $x \in [l, u]$ is a valid linear upper bound.
For initialization, we set the tangent point (either $\alpha_2$ or $\alpha_3$) to $t = 1/2 \cdot (l + u)$.
This choice leads to an area-optimal overapproximation~\cite{Biktairov2023sol}.

Since $\gelu$ is convex for $x \in [-\sqrt{2}, \sqrt{2}]$, we can use the secant as an upper bound and set $\alpha_1$ similarly as above.

In the cases \emph{non-convex middle increasing} and \emph{outer}, we can use the secant as a valid upper bound and the tangent at $t \in [l', u']$ for some $l \leq l', u' \leq u$ as a valid lower bound.
We employ the same procedure as Shi \textit{et al.}~\cite{Shi2025general} for finding $l', u'$ and clamp the value of $\alpha_1$ to these bounds if necessary.

For the cases \emph{middle decreasing II} and \emph{non-convex positive II}, the secant is a valid lower bound and the tangent at $t \in [l, u']$, respectively the tangent at $t \in [l', u]$ is a valid upper bound.
In these cases, we initialize $\alpha_2$ or $\alpha_3$ by computing $1/2 \cdot (l + u)$ and clamping to the valid bounds.
%\todo{why is this optimal?}.

Finally, for \emph{middle decreasing I} and \emph{non-convex positive I}, the tangent at $t \in [l', u']$ is a valid lower bound and the tangent at $t \in [l, u'']$ or respectively $t \in [l'', u]$ is a valid upper bound. 
As initialization for $\alpha_1, \alpha_2$ and $\alpha_3$, we choose again the mid-point of the interval $[l, u]$ clamped to the valid tangent point ranges.

\section{Additional Information on Differential Verification for Neural Networks}
\label{sec:diff-nn}

In this section, we provide information for how to instantiate the differential relaxation presented in Theorem~\ref{thm:p-act-relax} for the $\gelu$ function and report an optimization used to tighten the verified differential bounds.

\subsection{Bounds on the Derivative of $\gelu$}
\label{ssec:gelu-derivative-bounds}

Instantiating the differential relaxation for $p(x) - \sigma(x - \Delta)$ as described in Theorem~\ref{thm:p-act-relax} for $\sigma = \gelu$, requires computing bounds on the derivative $\gelu'(x)$ of the $\gelu$ function over an interval $x \in [l, u]$.

Since $\gelu(x)$ is concave for $D_1 = (-\infty, -\sqrt{2}]$, convex for $D_2 = [-\sqrt{2}, \sqrt{2}]$ and concave for $D_3 = [\sqrt{2}, \infty)$, its derivative is monotonically decreasing, monotonically increasing and monotonically decreasing over the respective domains.
Therefore, the extrema can only be attained at the boundaries of these regions:
\begin{align*}
    C = \bigcup_{i \in \{1,2,3\}} \left\{x \mid x \in \{l_i, u_i\},~ [l_i, u_i] = D_i \cap [l, u]\right\} \;,
\end{align*}
where we do not add the interval boundaries of degenerate intervals.
We can then just optimize over this finite set and use
\begin{align}
    \min / \max \left\{\gelu'(x) \mid x \in C\right\} \;.
\end{align}

\subsection{Increased Precision using Construction-time Verified Ranges}
\label{ssec:increased-precision-ranges}

Note that the formulas to compute the neuron-wise differences $\tilde{\vec{\Delta}}^{(l)}$ and $\vec{\Delta}^{(l)}$ in Section~\ref{sec:background_nnv} still partially depend on the values $\vec{y}_i^{(l-1)}$ and $\tilde{\vec{y}}_i^{(l)}$ of neurons in the individual networks $\net_1$ and $\net_2$.
Our differential verification approach extends \textsc{VeryDiff}~\cite{Teuber2025verydiff}, which computes neuron ranges based on zonotope propagation and is thus limited to linear relaxations with parallel slopes.
State-of-the-art NN verifiers like $\alpha$-CROWN~\cite{Xu21aCROWN} are known to compute tighter bounds for ranges in single neural networks.
If such a NN verifier was used to compute the verified pre-activation bounds of the polynomial network, we can use the tightened lower and upper bounds
\begin{align}
    l &= \max\{l_{zono}, l_p\}, && u = \min\{u_{zono}, u_p\} \;,
\end{align}
where $[l_{zono}, u_{zono}]$ are the ranges computed by \textsc{VeryDiff} and $[l_p, u_p]$ are the verified ranges for the polynomial network.

\section{Additional Evaluation Results}
\label{sec:additional-results}
\FloatBarrier

%\todo[inline]{If we desperately need space, remove this and put it in the code repo.}

In this appendix, we provide a more detailed description of our benchmarks and base neural networks.
Furthermore, we present additional evaluation results.

% \subsection{Experiment Settings}
% \label{ssec:experiment-settings}

% We implemented our approach for certified design of polynomial NNs in Julia~\cite{bezanson2017julia} on top of the differential verifier \textsc{VeryDiff}~\cite{Teuber2025verydiff}.
% Execution of polynomial NNs under CKKS is based on OpenFHE~\cite{OpenFHE}.

% To compute verified pre-activation ranges, we leverage the state-of-the-art NN verifier $\alpha$-CROWN~\cite{Xu21aCROWN}.

% We use a piecewise-polynomial overapproximation $q(x)$ for the $\gelu$ activation function with verified error $|q(x) - \gelu(x)| \leq 10^{-10}$.
% The overapproximation uses $6$ polynomial pieces of degree $15$ in addition to the linear asymptotes.

\subsection{Benchmarks}
\label{ssec:benchmarks}

In the following, we give a brief description of each dataset and the associated machine learning task. %and the training procedure of the neural networks for this benchmark.
%The network architectures for each benchmark are summarized in Table~\ref{tab:architecture}.
%All of the networks listed in this table are later used as base networks that are approximated by polynomial networks.
As prior work~\cite{Kern2025fhe} has already shown that training the base networks with $L_1$-regularization significantly reduces the polynomial approximation error, we trained all NNs with $L_1$-penalty (see Table~\ref{tab:large-scale}).

We run our approach on benchmarks related to classification tasks to evaluate the difference in accuracy between base and polynomial network, and also evaluate on regression tasks, where the verified error bounds are a more meaningful measure for comparison between the base and the polynomial network.

\subsubsection{Classification Benchmarks}\hfill

\textbf{HELOC.}
The \emph{home equity line of credit} dataset~\cite{FICOChallenge} contains $23$ features used to predict if a person is credit-worthy.
Following the procedure of prior work on certified construction of polynomial networks~\cite{Kern2025fhe}, we normalized the features to $[0,1]$.
We also use the $\relu$ neural networks provided by them.
Neural networks with $\gelu$ activation are obtained by training directly on the HELOC dataset.

\textbf{MNIST.}
Neural networks trained on the MNIST dataset~\cite{Lecun1998mnist} aim to correctly classify the handwritten digit shown on $28 \times 28$ pixel grayscale images.
Pixel values are normalized to $[0, 1]$.
Since Kern \textit{et al.}~\cite{Kern2025fhe} also evaluated on MNIST, we simply reuse their $\relu$ networks.
Additionally, we train networks of the same architecture with $\gelu$ activation on the MNIST dataset to obtain base networks with this activation function.

\textbf{CIFAR10.}
The task associated with the CIFAR10 dataset~\cite{krizhevsky2009cifar} is to assign the correct labels to images represented by inputs of dimension $3 \times 32 \times 32$, where each dimension can attain values in $[0,1]$.
The dataset is a popular benchmark used to evaluate FHE neural networks~\cite{Rovida2025fhe,Ebel2025orion}.
However, those approaches rely on sampled pre-activation bounds.
To achieve good accuracy of the polynomial networks -- despite wide verified bounds due to overapproximation incurred for large networks by NN verifiers -- we train two convolutional networks of different sizes ourselves.
The difference in performance between the small and the large network can then be used to judge scalability of our verified approach.

\subsubsection{Regression Benchmarks}\hfill

\textbf{Residential Electricity Forecasting (EC).} 
Prior research on FHE used electricity demand forecasting as a use case \cite{Bos2017energy}. In such setting, the smart meter associated to a residential building would send encrypted consumption data to a centralised server and get a short-term (half-hour) prediction of the future load. Similar to \cite{Bos2017energy}, we train a network that receives the past 24 hours of consumption data as input (averaged across $10$ homes), and predicts the next half-hour consumption. Additional input features include day of the year, day of the week, and time of the day. The dataset is real-world consumption data from a large-scale study from 2011, Ireland \cite{CER2011electricity}. As discovered by \cite{Veit2014forecasting}, a small fully-connected network with a single hidden layer suffice to maximise predictive accuracy; we choose a width of $64$ neurons.
% \todo[inline]{Be careful to present it as purely an equivalence problem, rather than an accuracy one (simple linear classifiers are quite good here).}

\textbf{NN4Sys.}
We use the \emph{learned index} models from the NN4Sys benchmark~\cite{Lin2024NN4Sys} used in the annual neural network verification competition \textsc{VNNComp}~\cite{vnncomp2025}.
The associated task is, to map a key (a scalar in $[0, 1]$) to an index (also a scalar in $[0, 1]$) in a database.
Since the available networks are not trained using $L_1$ regularization, we extract a training dataset from the input-output specifications that were used in the competition to train $relu$ and $gelu$ neural networks from scratch.
We use the same architecture as the non-$L_1$ regularized networks in the original NN4Sys benchmark.

\textbf{Collins RUL.}
The Collins RUL benchmark~\cite{Kirov2023CollinsRUL} was also used in the competition \textsc{VNNComp}~\cite{vnncomp2025}.
It consists of convolutional networks that process multivariate time series data (windows of $20$ time steps for $20$ sensors) to predict the remaining useful life of aircraft components.
Since the original networks provided by the benchmark are, again, not trained using $L_1$ regularization, we train $\relu$ and $\gelu$ networks of the same structure ourselves.
However, as the benchmark is only accompanied by a small test dataset, we train our networks to minimize the mean squared error $(\net_{L_1}(x) - \net(x))^2$ to the original network over random data (while monitoring the loss over the test dataset).

\subsection{Comparison to Prior Certified Approaches}
\label{ssec:comparison-prior-verified}

\begin{figure}[ht]
    \centering
    \includegraphics[width=\linewidth]{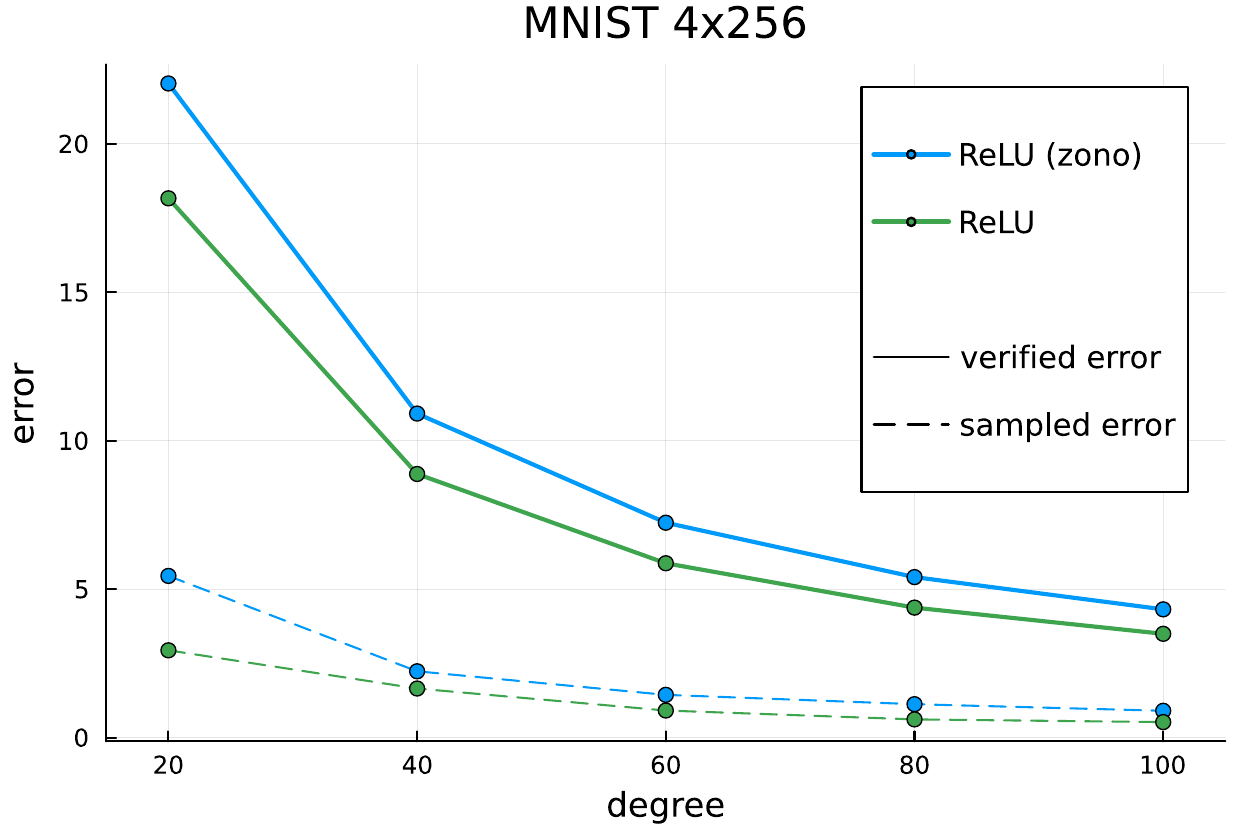}
    \caption{Comparison of our method to \textsc{ZonoPoly}~\cite{Kern2025fhe} for verified bounds and sampled values on $\|f_\pi(x) - f(x)\|_\infty$.}
    \label{fig:comparison_zonopoly}
\end{figure}

% \begin{figure*}
%     \centering
%     \includegraphics[width=\linewidth]{images/experiments/ablation_comparison_zonopoly.pdf}
%     \caption{Comparison our method and \textsc{ZonoPoly}~\cite{Kern2025fhe}  for verified and sampled bounds on $\| f_\pi(x) - f(x)\|_\infty$. \textsc{ZonoPoly} is only applicable to $\gelu$ and convolutional architectures due to our extensions.}
%     \label{fig:comparison_zonopoly}
% \end{figure*}

We also conducted experiments to compare our approach to $\textsc{ZonoPoly}$~\cite{Kern2025fhe} the prior state of the art in certified construction of polynomial neural networks.

Both our approach and \textsc{ZonoPoly} extend the zonotope-based differential verification tool \textsc{VeryDiff}~\cite{Teuber2025verydiff} to compute bounds on the output difference between original and polynomial neural networks.
However, our approach can leverage state of the art neural network verifiers (e.g.~\cite{Xu21aCROWN}) to compute bounds on the networks' preactivation ranges (see Section~\ref{ssec:range-analysis}), while \textsc{ZonoPoly} is restricted to less powerful zonotope-based methods for this step.

As \textsc{ZonoPoly} is only applicable to fully connected networks with $\relu$ activation function, we only compare performance on the MNIST $256 \times 4$ network with $\relu$ activation.
Starting from this base network, we constructed polynomial networks for a range of polynomial degrees and measured the verifiable bounds on the error $\| f_\pi(x) - f(x) \|_\infty$, as well as the sampled error over the MNIST dataset for each method.

The results shown in Figure~\ref{fig:comparison_zonopoly} demonstrate that our approach achieves consistently better error bounds than $\textsc{ZonoPoly}$ ($1.21 \times$ and $1.23 \times$ smaller verified error).

\FloatBarrier

% \begin{figure}
%     \centering
%     \includegraphics[width=0.7\linewidth]{images/cifar_large_sampled_nan.png}
%     \caption{Perturbed CIFAR 10.1 images on which the large CIFAR NN is vulnerable to an overflow events}
%     \label{fig:cifar-perturbations}
% \end{figure}
% \FloatBarrier

\section{Proofs}
\printProofs

\section{Network to JSON grammar}
\label{sec:network-to-json}
Our goal is to, given a polynomial network, being able to evaluate it under CKKS. In the literature we can find different examples of general purpose HE compilers, such as EVA~\cite{10.1145/3385412.3386023} or HECATE~\cite{9741265}. Unlike these frameworks, our input to the CKKS program is a JSON file containing a list of NN layers, with corresponding sizes and weights, as defined in the grammar in Figure \ref{fig:grammar}. Our goal is not to be as generic as the previously mentioned compilers.
Instead, we tailor the CKKS circuit construction to the NN use case and implemented an automatic FHE evaluation of each specified layer. Moreover, our evaluation takes as input coefficients of Chebyshev polynomials, one for each neuron, which enables the evaluation of overflow-free networks under CKKS using Algorithm \ref{alg:ps}. This evaluation of heterogenuous polynomials in a single layer is not, to the best of our knowledge, possible in any HE compiler.
%, as we simply require to evaluate in CKKS the networks crafted using our NN verification techniques, mostly implemented in Julia. For this reason we simply build, in OpenFHE, a program taking as input a JSON network and we implemented an automatic evaluation of each specified layer.
\subsubsection{Automatic evaluation of linear layers}
We implemented simple algorithms to evaluate linear layers. We first define the following encodings.
\begin{definition}[Repeated encoding]
    Given a vector ${\bm v} \in \mathbb{R}^n$, its \emph{repeated encoding} is a vector ${\bm v}' \in \mathbb{R}^{n^2}$ where ${\bm v}$ is repeated $n$ times. Visually:
    \[
    {\bm v}' := (v_1, v_2, \dots, v_n \,|\, v_1, v_2, \dots, v_n\,|\, \dots) \in \mathbb{R}^{n^2}.
    \]
\end{definition}
\begin{definition}[Expanded encoding]
    Given a vector ${\bm v} \in \mathbb{R}^n$, its \emph{expanded encoding} is a vector ${\bm v}' \in \mathbb{R}^{n^2}$ where each component $v_i$ of ${\bm v}$ is repeated $n$ times. Visually:
    \[
    {\bm v}' := (v_1, v_1, \dots, v_1 \,|\, v_2, v_2, \dots, v_2\,|\, \dots) \in \mathbb{R}^{n^2}.
    \]
\end{definition}
We use two different input encodings as we employ two different algorithms to evaluate the typical $W{\bm x} + {\bm b}$ transformation, depending on how the input is encoded (Algorithm \ref{alg:linear1} and \ref{alg:linear2}).
\begin{algorithm}[h]
\caption{Linear layer with expanded input}
\label{alg:linear1}
\begin{algorithmic}[1]
\Require Expanded ciphertext ${\bm c} \in \mathbb{R}^{n^2}$, column-wise weight matrix ${\bm w} \in \mathbb{R}^{n^2}$, repeated bias ${\bm b} \in \mathbb{R}^{n^2}$, dimension $n$
\Ensure Repeated encoding of $W{\bm x} + {\bm b} \in \mathbb{R}^{n^2}$
\State ${\bm r} := {\bm c} \odot {\bm w}$ \Comment{\it Component-wise multiplication}
\For{$i = 0, 1, \dots, \lfloor \log_2 n \rfloor - 1$}
    \State ${\bm r} := {\bm r} + \mathrm{rot}({\bm r}, n \cdot 2^i)$
\EndFor
\State \Return ${\bm r}$ + ${\bm b}$
\end{algorithmic}
\end{algorithm}
\begin{algorithm}[h]
\caption{Linear layer with repeated input}
\label{alg:linear2}
\begin{algorithmic}[1]
\Require Repeated ciphertext ${\bm c} \in \mathbb{R}^{n^2}$, row-wise weight matrix ${\bm w} \in \mathbb{R}^{n^2}$, expanded bias ${\bm b} \in \mathbb{R}^{n^2}$, dimension $n$
\Ensure Expanded encoding of $W{\bm x} + {\bm b} \in \mathbb{R}^{n^2}$
\State ${\bm r} \leftarrow {\bm c} \odot {\bm w}$
\For{$i = 0, 1, \dots, \lfloor \log_2 n \rfloor - 1$}
    \State ${\bm r} \leftarrow {\bm r} + \mathrm{rot}({\bm r},\, 2^i)$
\EndFor
\State ${\bm m} \leftarrow (1, \underbrace{0, \dots, 0}_{n-1}, 1, \underbrace{0, \dots, 0}_{n-1}, \dots) \in \mathbb{R}^{n^2}$ \Comment{\it Mask correct slots}
\State ${\bm r} \leftarrow {\bm r} \odot {\bm m}$
\For{$i = 0, 1, \dots, \lfloor \log_2 n \rfloor - 1$} \Comment{\it Repeat along empty slots}
    \State ${\bm r} \leftarrow {\bm r} + \mathrm{rot}({\bm r},\, -2^i)$
\EndFor
\State \Return ${\bm r} + {\bm b}$
\end{algorithmic}
\end{algorithm}

\label{sec:ckks_nn_translation}
%\todo[inline]{Final thing to write}
%\todo[inline]{Describe the vec x mat algorithms and the convolutions}
%\todo[inline]{Here we describe the software pipeline, include a specification of the PS algorithm, extended from ctxt*const to ctxt*ptxt}
For convolutions, we used ad-hoc source codes as their evaluation is harder to automate. However, our framework is flexible enough to allow for the automatic evaluation of convolutions by simply defining a function that handles the $\langle \text{conv-layer}\rangle$ tuple in the grammar presented in Figure \ref{fig:grammar}.
\subsubsection{Automatic evaluation of Chebyshev layers.}
The evaluation of Chebyshev polynomials follows the footprint firsly introduced in \cite{CCS19} --- which in turn implements a version of the Paterson-Stockmeyer algorithm \cite{PS73} adapted for Chebyshev bases --- however our approach is slightly more general as in \cite{CCS19} they assume to evaluate a single polynomial in all the slots of a ciphertext. On the other hand, in our construction we require to evaluate a different polynomial for each slot. Luckily, we can use the strategy described in \cite[Section 3.1]{CCS19} and replace the usage of constants coefficents $c_i$ to packed coefficients in a plaintext as ${\bm c}_i$. The following corollary easily follows.
\begin{corollary}[\mbox{generalizes \cite[Theorem 1]{CCS19}}]
    There exists an algorithm in CKKS to evaluate $d$ polynomials of degree $n$ given in Chebyshev base with $\sqrt{2n} + O(\log(n))$ non-scalar multiplications and $O(n)$ scalar multiplications, with $d$ being the number of available slots in a CKKS ciphertext.\label{cor:cheby}
\end{corollary}
Refer to \ref{alg:ps}.
We remark that the cost of evaluating a single polynomial in all slots or a different polynomial, one for each slot, is the same as constants are typically encoded in plaintexts in order to be compatible with the CKKS multiplication\footnote{Minor differences might occur in practice depending on how the multiplication algorithms are implemented for constants vs. plaintexts, but we can ignore them}.
We therefore used such algorithm to allow for the evaluation of different polynomials in different slots.
Since the original OpenFHE implementation of CKKS does not provide such an algorithm, we used a custom fork that enables the usage of the \texttt{EvalChebyshevSeriesPSBatch}.

\subsubsection{Building the grammar}
Given all the just presented approaches, we define a grammar in Figure \ref{fig:grammar} that will be the input to our CKKS program. In practice, given a polynomial network, we export it into a JSON file following the grammar, examples are given in the \texttt{networks/paper} folder of our public repository. Notice that, in its current state, che convolutional layer is still unsupported --- and ad-hoc code must be written --- as we leave the implementation of arbitrary convolution under CKKS as future work.

\subsubsection{Automatic parameters selection}
Given a JSON network, the program adapts the choice of the underlying CKKS scheme by first computing the required multiplicative depth. In practice, Algorithm \ref{alg:linear1} requires 1 multiplications and Algorithm \ref{alg:linear2} requires 2. On the other hand, the depth required by the evaluation of Chebyshev polynomials (Corollary \ref{cor:cheby}) depends on the degree as shown in Table \ref{tab:mult-depth-cheby}.
\begin{table}[h]
\centering
\caption{Multiplicative depth by degree using the Paterson-Stockmeyer algorithm \cite{PS73,CCS19}\vspace*{1mm}}
\label{tab:mult-depth-cheby}
\renewcommand{\arraystretch}{1.1}
\begin{tabular}{cc}
\textbf{Degree} & \textbf{Multiplicative Depth} \\
\hline
3--5    & 4  \\
6--13   & 5  \\
14--27  & 6  \\
28--59  & 7  \\
60--119 & 8  \\
120--247 & 9  \\
\end{tabular}
\end{table}
Given the overall multiplicative depth, along with a specified precision parameters ($\Delta, q_i$), the program fixes the minimum required ring size $N$ that ensures at least 128-bits of security according to \cite{FheGuidelines}.

\section{SIMD Paterson-Stockmeyer}
In this section we simply provide a formal definition of our modified version of the Paterson--Stockmeyer algorithm that supports the evaluation of $s$ polynomials in the $s$ different slots of a CKKS ciphertext (Algorithm \ref{alg:ps}).
\begin{algorithm}[h]
\caption{SIMD Paterson--Stockmeyer algorithm}
\label{alg:ps}
\begin{algorithmic}[1]
\Require{Set of Chebyshev coefficients $({\bm c}_0, {\bm c}_1, \ldots, {\bm c}_s)$ of equal degree, encrypted vector ${\bm u}$ of size $s$}
\Ensure{$p_i(u_i) = \displaystyle\sum_{j=0}^{n} (c_i)_j\, T_j(u_i)$}, for each $0 \le i < s$, with $s$ being the number of CKKS slots.\vspace*{2mm}

\Statex \emph{\small Note: the notation used within the algorithm is local and does not coincide with the notation used elsewhere}
\vspace*{2mm}

\State
  Find positive integers $k, m$ such that $k \approx \sqrt{n/2}$ and $k(2^{m}-1) > n$.

\State
  $\widetilde{p}(x) \gets p(x) + T_{k(2^{m}-1)}(x)$

\State
  Compute $\mathcal{B}_i \gets \bigl(T_1(u_i),\, T_2(u_i),\, \ldots,\, T_k(u_i)\bigr)$ for each $i < s$ via the recurrence
  \[
  \begin{split}
    T_0(u_i)&=1,\quad T_1(u_i)=u_i,\quad\\
    T_j(u_i) &= 2u_i\cdot T_{j-1}(u_i) - T_{j-2}(u_i),
    \end{split}
  \]
  compute $\mathcal{G}_i \gets \bigl(T_k(u_i),\, T_{2k}(u_i),\, T_{4k}(u_i),\, \ldots)$
  via the doubling formula $T_{2\ell}(u_i) = 2\,T_{\ell}(u_i)^{2} - 1$. 

\State
  Apply \cite[Lemma~1]{CCS19} to define $q_i(x), r_i(x)$ as
  \[
    \widetilde{p}_i(x) \gets T_{k\cdot 2^{m-1}}(x)\cdot q_i(x) + r_i(x),
  \]
  with $\deg q_i = k(2^{m-1}-1)$ and $\deg r_i < k\cdot 2^{m-1}$.

\State
  $\widetilde{r}_i(x) \gets r_i(x) - T_{k(2^{m-1}-1)}(x)$.

\State
  Apply \cite[Lemma~1]{CCS19} again to define $d_i(x)$ and $s_i(x)$ such that
$\widetilde{r}_i(x) = d_i(x)\cdot q_i(x) + s_i(x)$.

\State
  Evaluate $d_i(u_i)$ using the Chebyshev coefficients of ${\bm c}_i$ and the
  precomputed values $\mathcal{B}_i$

\State
  $\widetilde{s}_i(x) \gets s_i(x) + T_{k(2^{m-1}-1)}(x)$.

\State
  Recursively evaluate $q_i(u_i)$ and $\widetilde{s}_i(u_i)$ at $u_i$ (skip Steps~1--3; reuse $\mathcal{B}$ and $\mathcal{G}$).

\State
  Compute
  \[
    \widetilde{p}_i(u_i) \gets
    \bigl(T_{k\cdot 2^{m-1}}(u_i) + d_i(u_i)\bigr)\cdot q_i(u_i) + \widetilde{s}_i(u_i),
  \]
  where $T_{k\cdot 2^{m-1}}(u_i)$ is the last entry of $\mathcal{G}_i$.

\State
  Compute $T_{k(2^{m}-1)}(u_i)$ from $\mathcal{G}_i$ using the iterative formula
  derived from the Chebyshev product identity
  $T_a\,T_b = \tfrac{1}{2}\bigl(T_{a+b}+T_{|a-b|}\bigr)$:
  \[
  \begin{split}
    V_1 &\gets T_k(u_i),\qquad\\
    V_j &\gets 2\,V_{j-1}\cdot T_{k\cdot 2^{j-1}}(u_i) - T_k(u_i),
    \quad j=2,\ldots,m;
    \end{split}
  \]

\State \Return $p_i(u_i) \gets \widetilde{p}_i(u_i) - V_m$
\end{algorithmic}
\end{algorithm}

Notice that the just presented algorithm is simply an extension of the original one, taking as input sets of coefficients instead of a single set. Also notice that, since the flow %\todo{what is flux? LR: i changed to flux of execution :p} 
of execution of the algorithm heavily depends on the degree(s) of the polynomial(s), all the different degrees of ${\bm c}_i$ must be equal in our instantiation. Similarly, the integers $k, m$ must also be the same for all the different polynomials.

\begin{figure*}[h]
    \centering
\begin{align*}
\langle\text{root}\rangle &::= \texttt{[} \ \langle\text{layer-list}\rangle \ \texttt{]} \\
\langle\text{layer-list}\rangle &::= \langle\text{layer}\rangle \ \langle\text{layer-rest}\rangle \mid \varepsilon \\
\langle\text{layer-rest}\rangle &::= \texttt{,} \ \langle\text{layer}\rangle \ \langle\text{layer-rest}\rangle \mid \varepsilon \\
\langle\text{layer}\rangle &::= \langle\text{linear-layer}\rangle \mid \langle\text{chebyshev-layer}\rangle \mid \langle\text{conv-layer}\rangle \\
\\
\langle\text{linear-layer}\rangle &::= \texttt{[} \ \texttt{"linear"} \texttt{,} \ \langle\text{num-array}\rangle \texttt{,} \ \langle\text{num-array}\rangle \texttt{,} \ \langle\text{shape}\rangle \ \texttt{]} \\
\\
\langle\text{chebyshev-layer}\rangle &::= \texttt{[} \ \texttt{"chebyshev"} \texttt{,} \ \langle\text{num-array}\rangle \texttt{,} \ \langle\text{num-array}\rangle \texttt{,} \ \langle\text{num-array}\rangle \texttt{,} \ \langle\text{shape}\rangle \ \texttt{]} \\
\\
\langle\text{conv-layer}\rangle &::= \texttt{[} \ \texttt{"conv"} \texttt{,} \ \langle\text{shape}\rangle \texttt{,} \ \langle\text{integer}\rangle \texttt{,} \ \langle\text{integer}\rangle \texttt{,} \\
&\phantom{::=} \ \langle\text{num-array}\rangle \texttt{,} \ \langle\text{num-array}\rangle \texttt{,} \ \langle\text{shape}\rangle \texttt{,} \ \langle\text{shape2}\rangle \ \texttt{]} \\
\\
\langle\text{num-array}\rangle &::= \texttt{[} \ \langle\text{number-list}\rangle \ \texttt{]} \mid \texttt{[]} \\
\langle\text{number-list}\rangle &::= \langle\text{number}\rangle \mid \langle\text{number}\rangle \ \texttt{,} \ \langle\text{number-list}\rangle \\
\langle\text{shape}\rangle &::= \texttt{[} \ \langle\text{number}\rangle \texttt{,} \ \langle\text{number}\rangle \ \texttt{]} \\
\langle\text{shape2}\rangle &::= \texttt{[} \ \langle\text{number}\rangle \texttt{,} \ \langle\text{number}\rangle \texttt{,} \ \langle\text{number}\rangle \texttt{,} \ \langle\text{number}\rangle \ \texttt{]} \\
\\
\langle\text{number}\rangle &::= \langle\text{integer}\rangle \mid \langle\text{float}\rangle \\
\langle\text{integer}\rangle &::= \langle\text{digits}\rangle \mid \texttt{-} \ \langle\text{digits}\rangle \\
\langle\text{float}\rangle &::= \langle\text{digits}\rangle \ \texttt{.} \ \langle\text{digits}\rangle \mid \texttt{-} \ \langle\text{digits}\rangle \ \texttt{.} \ \langle\text{digits}\rangle \\
\langle\text{digits}\rangle &::= \langle\text{digit}\rangle \mid \langle\text{digit}\rangle \ \langle\text{digits}\rangle \\
\langle\text{digit}\rangle &::= \texttt{0} \mid \texttt{1} \mid \texttt{2} \mid \texttt{3} \mid \texttt{4} \mid \texttt{5} \mid \texttt{6} \mid \texttt{7} \mid \texttt{8} \mid \texttt{9}
\end{align*}
    \caption{Grammar used to build a JSON file executable in our CKKS circuit}
    \label{fig:grammar}
    \end{figure*}

% that's all folks
\end{document}